\newif\ifreport\reporttrue
\documentclass[11pt, draftclsnofoot,onecolumn]{IEEEtran}
\usepackage{setspace}
\usepackage{cite}
\usepackage{graphicx}
\usepackage{caption}
\usepackage{subcaption}
\usepackage[cmex10]{amsmath}
\usepackage{mathtools}
\usepackage{amsthm}
\usepackage{amsfonts}
\usepackage{amssymb}
\usepackage[lined, boxed, linesnumbered, ruled]{algorithm2e}
\usepackage{bm,xstring}
\usepackage{fixltx2e}
\usepackage{tikz}
\usetikzlibrary{decorations.pathreplacing}

\usepackage{color}

\newcommand{\ignore}[1]{}
\usepackage[singlelinecheck=false]{caption}

\newtheorem{lemma}{Lemma}
\newtheorem{proposition}{Proposition}

\newtheorem{theorem}{Theorem}
\newtheorem{corollary}{Corollary}
\theoremstyle{definition}

\newtheorem{definition}{Definition}

\begin{document}




\title{Near Delay-Optimal Scheduling of Batch Jobs in Multi-Server  Systems}

\author{~\\Yin Sun, C. Emre Koksal, and Ness B. Shroff \\
~\\
February 2, 2017\\

\thanks{Yin Sun and C. Emre Koksal are with the Department of Electrical and Computer Engineering, the Ohio State University, Columbus, OH. Email: sunyin02@gmail.com, koksal.2@osu.edu.

Ness B. Shroff is with the Departments of Electrical and Computer Engineering and Computer Science and Engineering, the Ohio State University, Columbus, OH. Email: shroff.11@osu.edu.}
}

\maketitle
\begin{abstract}
We study a class of  scheduling problems, 
where each job is divided into a batch of unit-size tasks and these tasks can be executed in parallel on multiple servers with New-Better-than-Used (NBU) service time distributions. While many delay optimality results are available for single-server queueing systems, generalizing these results to the multi-server case has been challenging. 
This motivated us to investigate near delay-optimal scheduling of batch jobs in multi-server queueing systems. 
We consider three low-complexity scheduling policies: the Fewest Unassigned Tasks first (FUT) policy, the Earliest Due Date first (EDD) policy, and the First-Come, First-Served  (FCFS) policy. 
We prove that for arbitrary number, batch sizes, arrival times, and due times of the jobs, these scheduling policies are near delay-optimal in stochastic ordering for minimizing three classes of delay metrics among all causal and non-preemptive policies.
In particular, the FUT policy is within a constant additive delay gap from the optimum  for minimizing the mean average delay, and  
the FCFS policy within twice of the optimum for minimizing the mean maximum delay and the mean $p$-norm of delay. 
The key proof tools are several novel sample-path orderings, which can be used to compare the sample-path delay of different policies in a near-optimal sense. 

\end{abstract}
\newpage

\section{Introduction}

Achieving low service latency is imperative in cloud computing and data center systems. 
As the user population and data volume scale up, the running times of service jobs grow quickly. 
In  practice, 
long-running jobs are broken into a batch of short tasks which can be executed separately by a number of parallel servers  \cite{mapreduce}. We study the scheduling of such batch jobs in multi-server queueing systems. 
The jobs arrive over time according to a  general arrival process, where the number, batch sizes, arrival times, and due times of the jobs are \emph{arbitrarily} given. A scheduler determines how to assign the tasks of incoming jobs to the servers, based on the casual information (the history and current information) of the system. 
We assume that each task is of one unit of work and can be flexibly assigned to any server. The service times of the tasks follow New-Better-than-Used (NBU) distributions, and are \emph{independent} across the servers and \emph{i.i.d.} across the tasks assigned to the same server. 
Our goal is to seek for low-complexity scheduling policies that optimize the delay performance of the jobs.

Many delay-optimal scheduling results have been established in single-server queueing systems. In the deterministic scheduling models, the preemptive Shortest Remaining Processing Time first (SRPT) policy, the preemptive Earliest Due Date first (EDD) policy, and the First-Come, First-Served  (FCFS) policy were proven to minimize the average delay \cite{Schrage68,Smith78}, maximum lateness \cite{Jackson55}, and maximum delay \cite{Baccelli:1993}, respectively. In addition,  index policies are shown to be delay-optimal in many stochastic scheduling problems, e.g., \cite{Smith56,Rothkopf66,Gittins79banditprocesses,Klimov74,Gittins2011}. 

Unfortunately, generalizing these delay-optimal results to the parallel-server case is quite challenging. In particular, minimizing the average delay in deterministic scheduling problems with more than one servers is $NP$-hard \cite{Leonardi:1997}. Similarly, delay-optimal stochastic scheduling in multi-class multi-server queueing systems is deemed to be notoriously difficult \cite{Weiss:1992,Weiss:1995}. Prior attempts on solving the delay-optimal stochastic scheduling problem have met with little success, except in some limiting regions such as large system limits, e.g., \cite{Ying2015}, and heavy traffic limits, e.g.,\cite{Dacre1999, Stolyar_heavy2004}. However, these results may not apply either outside of these limiting regions, or for a non-stationary job arrival process under which the steady state distribution of the system does not exist.

\subsection{Our Contributions}
Because of the difficulty in establishing delay optimality in the considered queueing systems, we investigate near delay-optimal scheduling in this paper.
We prove that 
for arbitrarily given job parameters (including the number,  batch sizes, arrival times, and due times of the jobs), the Fewest Unassigned Tasks first (FUT) policy, the EDD  policy, and the  FCFS  policy are near delay-optimal in  stochastic ordering
for minimizing average delay, maximum lateness, and maximum delay, respectively, among all causal and non-preemptive\footnote{We consider task-level non-preemptive policies: Processing of a task cannot be interrupted until the task is completed. However, after completing a task from one job, a server can switch to process another task from any job.} policies. If the job parameters satisfy some additional conditions, we can further show that these policies are near delay-optimal in stochastic ordering for minimizing three  general classes of delay metrics. In particular, we  show that the FUT policy is within a constant additive delay gap from the optimum  for minimizing the mean average delay (Theorem \ref{lem7_NBU})\footnote{As we will see in Section \ref{sec_analysis}, the FUT policy is a nice approximation of the Shortest Remaining Processing Time (SRPT) first policy \cite{Schrage68,Smith78}.}; and  
the FCFS policy within twice of the optimum for minimizing the mean maximum delay and mean $p$-norm of delay (Theorems \ref{thm_max_delay_ratio} and \ref{thm_max_delay_ratio_1}). 

We develop a unified sample-path method to establish these results, where the key proof tools are several novel sample-path orderings for comparing the delay performance of different policies in a near-optimal sense. These sample-path orderings are very general because they do not need to specify the queueing system model, and hence can be potentially used for establishing near delay optimality results in other queueing systems.

\subsection{Organization of the Paper}
We describe the system model and problem formulation in Section \ref{sec_model}, together with the notations that we will use throughout the paper. Some near delay optimality results are introduced in Section \ref{sec_analysis}, which are proven in Section \ref{sec_proofmain}.
Finally, some conclusion discussions are presented in Section \ref{sec_conclusion}.

\section{Model and Formulation}\label{sec_model}
\subsection{Notations and Definitions}
We will use lower case letters such as $x$ and $\bm{x}$, respectively, to represent deterministic scalars and vectors.
In the vector case, a subscript will index the components of a vector, such as $x_i$.
We use $x_{[i]}$ and $x_{(i)}$, respectively, to denote the $i$-th largest and the $i$-th smallest components of $\bm{x}$.  
For any $n$-dimensional vector $\bm{x}$, let $\bm{x}_{\uparrow}=(x_{(1)},\ldots,x_{(n)})$ 
denote the increasing 
rearrangement of $\bm{x}$.  Let $\bm{0}$
denote the vector 
with all 0 
components.

Random variables and vectors will be denoted by upper case letters such as $X$ and $\bf{X}$, respectively, with the subscripts and superscripts following the same conventions as in the deterministic case. 
Throughout the paper, ``increasing/decreasing'' and ``convex/concave'' are used in the non-strict sense. LHS and RHS denote, respectively, ``left-hand side'' and ``right-hand side''.

For any $n$-dimensional vectors $\bm{x}$ and $\bm{y}$, the elementwise vector ordering $x_i\leq y_i$, $i=1,\ldots,n$, is denoted by $\bm{x} \leq \bm{y}$. Further, $\bm{x}$ is said to be \emph{majorized} by $\bm{y}$, denoted by $\bm{x}\prec\bm{y}$, if (i) $\sum_{i=1}^j x_{[i]} \leq \sum_{i=1}^j y_{[i]}$, $j=1,\ldots,n-1$ and (ii) $\sum_{i=1}^n x_{[i]} = \sum_{i=1}^n y_{[i]}$ \cite{Marshall2011}. In addition, $\bm{x}$ is said to be  \emph{weakly majorized by $\bm{y}$ from below}, denoted by $\bm{x}\prec_{\text{w}}\bm{y}$, if $\sum_{i=1}^j x_{[i]} \leq \sum_{i=1}^j y_{[i]}$, $j=1,\ldots,n$; $\bm{x}$ is said to be  \emph{weakly majorized by $\bm{y}$ from above}, denoted by $\bm{x}\prec^{\text{w}}\bm{y}$, if $\sum_{i=1}^j x_{(i)} \geq \sum_{i=1}^j y_{(i)}$, $j=1,\ldots,n$ \cite{Marshall2011}.
A function that preserves the majorization order is called a Schur convex function. Specifically, $f: \mathbb{R}^n\rightarrow \mathbb{R}$ is termed \emph{Schur convex} if $f(\bm{x})\leq f(\bm{y})$ for all $\bm{x}\prec\bm{y}$ \cite{Marshall2011}. A function $f: \mathbb{R}^n\rightarrow \mathbb{R}$ is termed \emph{symmetric} if $f(\bm{x})= f(\bm{x}_{\uparrow})$ for all $\bm{x}$. A function $f: \mathbb{R}^n\rightarrow \mathbb{R}$ is termed \emph{sub-additive} if $f(\bm{x}+\bm{y})\leq f(\bm{x})+f(\bm{y})$ for all $\bm{x},\bm{y}$. The composition of functions $\phi$ and $f$ is denoted by $\phi \circ f (\bm{x}) = \phi(f (\bm{x}))$. Define $x\wedge y=\min\{x,y\}$.

Let $\mathcal{A}$ and $\mathcal{S}$ denote sets and events, with $|\mathcal{S}|$ denoting the cardinality of $\mathcal{S}$.
For all random variable ${X}$ and event $\mathcal{A}$, let $[{X}|\mathcal{A}]$ denote a random variable whose distribution is identical to the conditional distribution of ${X}$ for given $\mathcal{A}$. A random variable ${X}$ is said to be stochastically smaller than another random variable ${Y}$, denoted by ${X}\leq_{\text{st}}{Y}$, if $\Pr({X}>x) \leq \Pr({Y}>x)$ for all~$x\in \mathbb{R}$.
A set $\mathcal{U} \subseteq \mathbb{R}^n$ is called \emph{upper}, if $\bm{y} \in \mathcal{U}$ whenever $\bm{y}\geq \bm{x}$ and $\bm{x} \in \mathcal{U}$. 
A random vector $\bm{X}$ is said to be \emph{stochastically smaller} than another random vector $\bm{Y}$, denoted by $\bm{X}\leq_{\text{st}}\bm{Y}$, if $\Pr(\bm{X}\in \mathcal{U}) \leq \Pr(\bm{Y}\in \mathcal{U})$ for all upper sets ~$\mathcal{U}\subseteq \mathbb{R}^n$. If $\bm{X}\leq_{\text{st}}\bm{Y}$ and $\bm{X}\geq_{\text{st}}\bm{Y}$, then $\bm{X}$ and $\bm{Y}$ follow the same distribution, denoted by $\bm{X}=_{\text{st}}\bm{Y}$. We remark that $\bm{X}\leq_{\text{st}}\bm{Y}$  if, and only if,
\begin{align}\label{eq_order_property}
 \mathbb{E}[\phi(\bm{X})] \leq \mathbb{E}[\phi(\bm{Y})]
\end{align}
holds for all increasing $\phi: \mathbb{R}^n\rightarrow \mathbb{R}$ for which the expectations in \eqref{eq_order_property} exist.

\begin{figure}
\centering
\includegraphics[width=0.5\textwidth]{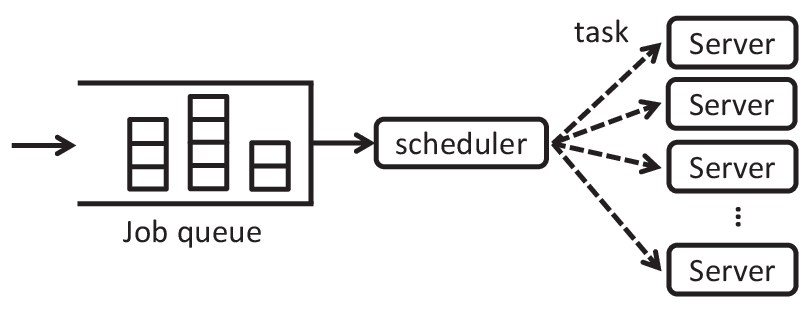}
\caption{A centralized queueing system, where a scheduler assigns the tasks of arrived jobs to the servers.} \label{fig1model_central}
\end{figure}

\subsection{Queueing System Model}\label{sec:model}
Consider a queueing system with $m$ parallel servers, as shown in Fig. \ref{fig1model_central}, which starts to operate at time $t=0$. 
A sequence of $n$ jobs arrive at time instants $a_1,\ldots,$ $a_n$ and are stored in a job queue, where $n$ is an arbitrary positive integer, no matter finite or infinite, and $0=a_1\leq a_2\leq\cdots\leq a_n$. 
The $i$-th arrived job, called job $i$, brings with it a batch of $k_i$ tasks.
Each task is of one unit of work, and can be flexibly assigned to any server.  Job $i$ is completed when all $k_i$ tasks of job $i$ are completed. The maximum job size is\footnote{If $n\rightarrow\infty$, then the $\max$ operator in \eqref{eq_0} is replaced by $\sup$.}
\begin{align}\label{eq_0}
k_{\max} = \max_{i=1,\ldots,n}k_i.
\end{align}
\label{sec_Multivariate}


A scheduler assigns the tasks of arrived jobs to the servers over time. The service times of the tasks  are assumed to be \emph{independent} across the servers and \emph{i.i.d.} across the tasks assigned to the same server.
Let $X_l$ represent the random service time of server $l$. The service rate of server $l$ is $\mu_l= 1/\mathbb{E}[X_l]$, which may vary across the servers. Denote $\bm{X}=(X_1,\dots,X_m)$.
We consider the classes of NBU task service time distributions.
\begin{definition}
Consider a non-negative random variable $X$ with complementary cumulative distribution function (CDF)  $\bar{F}(x)=\Pr[X>x]$. Then, $X$ is \textbf{New-Better-than-Used (NBU)} if for all $t,\tau\geq 0$
\begin{eqnarray}\label{eq_NBU}
\bar{F}(\tau+t)\leq \bar{F}(\tau)\bar{F}(t).
\end{eqnarray}
\end{definition}

\subsection{Scheduling Policies}


A scheduling policy, denoted by $\pi$, determines the task assignments 
 in the system.
We consider the class of \textbf{causal} policies, in which scheduling decisions are made based on the history and current state of the system; the realization of task service time is unknown until the task is completed (unless the service time is deterministic). 
In practice, service preemption is costly and may leads to complexity and reliability issues \cite{Sparrow:2013,Borg2015}. Motivated by this, we assume that \textbf{preemption is not allowed}. Hence, if a server starts to process a task, it must complete this task before switching to process another task. 
We use $\Pi$ to denote the set of causal and non-preemptive policies. Let us define several sub-classes of policies within $\Pi$:

A policy is said to be \textbf{non-anticipative}, if the scheduling decisions are made without using any information about the future job arrivals; and is said to be \textbf{anticipative}, if it has access to the parameters $(a_i,{k}_i,d_i)$ of future arriving jobs. For periodic and pre-planned services, future job arrivals can be predicted in advance. To cover these scenarios, we abuse the definition of causal policies a bit and include anticipative policies into the policy space $\Pi$. However, the policies that we propose in this paper are all {non-anticipative}. 

A policy is said to be \textbf{work-conserving}, if no server is idle when there are tasks waiting to be processed.

The goal of this paper is to design low-complexity non-anticipative scheduling policies that are near delay-optimal among all policies in $\Pi$, even compared to the anticipative policies with knowledge about future arriving jobs.

\subsection{Delay Metrics}\label{sec:metrics}
Each job $i$ has a \textbf{due time} $d_i\in[0,\infty)$, also called {due date}, which is the time that job $i$ is promised to be completed \cite{michael2012book}. 
Completion of a job after its due time is allowed, but then a penalty is incurred. Hence, the due time can be considered as a soft deadline. 

For each job $i$, $C_i$ is the job completion time, $D_i= C_i-a_i$ is the delay, $L_i = C_i-d_i$ is the {lateness} after the due time $d_i$, and $T_i = \max[C_i-d_i,0]$ is the  {tardiness} (or positive lateness). 
Define the vectors $\bm{a} =(a_1,\ldots,a_n)$,
$\bm{d} =(d_1,\ldots,d_n)$, $\bm{C}=(C_1,$ $\ldots,C_n)$, $\bm{D}=(D_1,\ldots, D_n)$,  $\bm{L}=(L_1,\ldots,$ $ L_n)$, and $\bm{C}_{\uparrow}\! =\! (C_{(1)},\ldots,C_{(n)})$. 
Let $\bm{c}=(c_1,$ $\ldots,c_n)$ and $\bm{c}_{\uparrow}\! =\! (c_{(1)},\ldots,c_{(n)})$, respectively, denote the realizations of $\bm{C}$ and $\bm{C}_{\uparrow}$. All these quantities are functions of the  scheduling policy $\pi$.


%

Several important delay metrics are introduced in the following: For any policy $\pi$, the \textbf{average delay} ${D}_{\text{avg}}: \mathbb{R}^n\rightarrow \mathbb{R}$ is defined by\footnote{If $n\rightarrow\infty$, then a $\limsup$ operator is enforced on the RHS of \eqref{eq_2}, and the $\max$ operator in \eqref{eq_1} and \eqref{eq_11} is replaced by $\sup$.
}
\begin{align}
{D}_{\text{avg}}(\bm{C}(\pi))=\frac{1}{n}\sum_{i=1}^{n}\left[C_i(\pi)-a_i\right]. \label{eq_2}
\end{align}
The \textbf{maximum lateness}  ${L}_{\max}: \mathbb{R}^n\rightarrow \mathbb{R}$ is defined by
\begin{align}
L_{\max}(\bm{C}(\pi)) \!= \!\max_{i=1,2,\ldots,n} L_i(\pi)=\max_{i=1,2,\ldots,n} \left[C_i(\pi)-d_i\right],\label{eq_1}
\end{align}
if $d_i=a_i$, $L_{\max}$ reduces to the \textbf{maximum delay}  ${D}_{\max}: \mathbb{R}^n\rightarrow \mathbb{R}$, i.e.,
\begin{align}
D_{\max}(\bm{C}(\pi))\!=\! \max_{i=1,2,\ldots,n} D_i(\pi)= \max_{i=1,2,\ldots,n} \left[C_i(\pi)-a_i\right].\!\!\label{eq_11}
\end{align}
The \textbf{$p$-norm of delay} ${D}_{p}: \mathbb{R}^n\rightarrow \mathbb{R}$  is defined by
\begin{align}
D_p(\bm{C}(\pi)) = || \bm{D}(\pi) ||_p = \left[\sum_{i=1}^n D_i(\pi)^p\right]^{1/p},~ ~p\geq 1.
\end{align}


In general, a \textbf{delay metric} can be expressed as a function $f(\bm{C}(\pi))$ of the job completion time vector $\bm{C}(\pi)$, where $f: \mathbb{R}^n\rightarrow \mathbb{R}$ is increasing. In this paper, we consider three classes of delay metric functions:
\begin{align}
\mathcal{D}_{\text{sym}} &= \{f : f \text{ is symmetric and increasing}\},\nonumber\\
\mathcal{D}_{\text{Sch-1}} &= \{f:  f (\bm{x}+\bm{d}) \text{ is Schur convex and increasing in $\bm{x}$}\},\nonumber\\
\mathcal{D}_{\text{Sch-2}} &= \{f:  f (\bm{x}+\bm{a}) \text{ is Schur convex and increasing in $\bm{x}$}\}.\nonumber\end{align}
Consider two alternative expressions of the delay metric function $f$:
\begin{align}
&f_1(\bm{L}(\pi))=f(\bm{L}(\pi)+\bm{d})=f(\bm{C}(\pi)),\label{eq_expression_1}\\
&f_2(\bm{D}(\pi))=f(\bm{D}(\pi)+\bm{a})=f(\bm{C}(\pi)).\label{eq_expression_2}\end{align}
Then, $f_1(\bm{L}(\pi))$ is Schur convex in the lateness vector $\bm{L}(\pi)$ for each $f\in \mathcal{D}_{\text{Sch-1}}$, and $f_2(\bm{D}(\pi))$ is Schur convex in the delay vector $\bm{D}(\pi)$ for each $f\in \mathcal{D}_{\text{Sch-2}}$. Furthermore, every convex and symmetric function is Schur convex. Using these properties, it is easy to show that
\begin{align}
&{D}_{\text{avg}}\in\mathcal{D}_{\text{sym}}\cap \mathcal{D}_{\text{Sch-1}} \cap \mathcal{D}_{\text{Sch-2}},\nonumber \\
&L_{\max}\in\mathcal{D}_{\text{Sch-1}}, D_{\max}\in\mathcal{D}_{\text{Sch-2}}, D_p \in \mathcal{D}_{\text{Sch-2}}.\nonumber
\end{align}



\subsection{Near Delay Optimality}\label{sec_optimality}

Define $\mathcal{I}=\{n,(a_i,{k}_i,d_i)_{i=1}^n\}$ as the parameters of the jobs, which include the number, batch sizes, arrival times, and due times of the jobs.
The job parameters $\mathcal{I}$ and random task service times are determined by two external processes, which are \emph{mutually independent} and do not change according to the scheduling policy. 
For delay metric function $f$ and policy space $\Pi$, a policy $P\in \Pi$ is said to be \textbf{delay-optimal in stochastic ordering}, if for all $\pi\in \Pi$ and $\mathcal{I}$
\begin{align}\label{eq_optimal}
[f (\bm{C}(P))|P,\mathcal{I}]\leq_{\text{st}} [f (\bm{C}(\pi))|\pi,\mathcal{I}].
\end{align}
or equivalently, if for all $\mathcal{I}$
\begin{align}\label{eq_optimal1}
\mathbb{E}[\phi \circ f (\bm{C}(P))|P,\mathcal{I}]=\min_{\pi\in\Pi}\mathbb{E}[\phi \circ f (\bm{C}(\pi))|\pi,\mathcal{I}]
\end{align}
holds for all increasing function $\phi: \mathbb{R}\rightarrow \mathbb{R}$ for which the conditional expectations in \eqref{eq_optimal1} exist.
The equivalence between \eqref{eq_optimal} and \eqref{eq_optimal1} follows from \eqref{eq_order_property}. 
For notational simplicity, we will omit the mention of the condition that policy $\pi$ or policy $P$ is adopted in the system in the rest of the paper. Note that this definition of delay optimality is quite strong because the optimal policy $P$ needs to satisfy \eqref{eq_optimal} and \eqref{eq_optimal1} for \emph{arbitrarily given} job parameters  $\mathcal{I}$.


%

In multi-class multi-server queueing systems, delay optimality is extremely difficult to achieve,
even with respect to some definitions of delay optimality that are weaker than \eqref{eq_optimal}. 
In deterministic scheduling problems, it was shown that minimizing the average delay in multi-server queueing systems is $NP$-hard 
\cite{Leonardi:1997}. Similarly, delay-optimal stochastic scheduling in multi-server queueing systems with multiple job classes is deemed to be notoriously difficult \cite{Weiss:1992,Weiss:1995,Dacre1999}.
This motivated us to study whether there exist policies that can come close to delay optimality. We will show that this is possible in the following sense:

\ifreport
\begin{figure}
\centering
\includegraphics[width=0.55\textwidth]{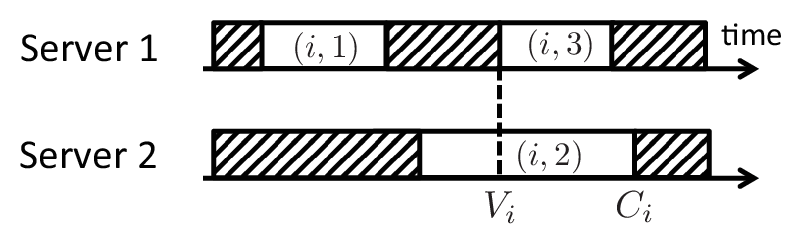} \caption{An illustration of $V_i$ and $C_i$. There are 2 servers, and job $i$ has 3 tasks denoted by $(i,1)$, $(i,2)$, $(i,3)$. Tasks $(i,1)$ and $(i,3)$ are assigned to Server 1, and Task $(i,2)$ is assigned to Server 2. By time $V_i$, all tasks of job $i$ have started service. By time $C_i$, all tasks of job $i$ have completed service. Therefore, $V_i\leq C_i$.}
\label{V_i}
\end{figure}
\fi

{Define $V_i$ as the earliest time that all tasks of job $i$ have entered the servers. In other words, all tasks of job $i$ are either completed or under service   at time $V_i$. One illustration of $V_i$ is provided in Fig. \ref{V_i}, from which it is easy to see 
\begin{align}\label{eq_V_i_small}
V_i\leq C_i.
\end{align}
Denote $\bm{V}=(V_1,\ldots, V_n)$, $\bm{V}_{\uparrow} \!=\! (V_{(1)},\ldots,V_{(n)})$. Let $\bm{v}=(v_1,\ldots, v_n)$ and $\bm{v}_{\uparrow} \!=\! (v_{(1)},\ldots,v_{(n)})$ be the realization of $\bm{V}$ and $\bm{V}_{\uparrow}$, respectively. All these quantities are functions of the  scheduling policy $\pi$.
A policy $P\in \Pi$ is said to be \textbf{near delay-optimal in stochastic ordering}, if for all $\pi\in \Pi$ and $\mathcal{I}$
\begin{align}\label{eq_nearoptimal}
[f (\bm{V}(P))|\mathcal{I}]\leq_{\text{st}} [f (\bm{C}(\pi))|\mathcal{I}],
\end{align}
or equivalently, if for all $\mathcal{I}$
\begin{align}\label{eq_nearoptimal1}
\mathbb{E}[\phi \!\circ\! f (\bm{V}(P))|\mathcal{I}] \!\leq\! \min_{\pi\in\Pi}\mathbb{E}[\phi \!\circ\! f (\bm{C}(\pi))|\mathcal{I}] \!\leq\!\mathbb{E}[\phi \!\circ\! f (\bm{C}(P))|\mathcal{I}]
\end{align}
holds for all increasing function $\phi$ for which the conditional expectations in \eqref{eq_nearoptimal1} exist. There exist many ways to approximate \eqref{eq_optimal} and \eqref{eq_optimal1}, and obtain various forms of near delay optimality.
We find that the form in \eqref{eq_nearoptimal} and \eqref{eq_nearoptimal1} is convenient, because it is analytically 
provable and leads to tight sub-optimal delay gap, as we will see in the subsequent sections.

\begin{algorithm}[h]
\SetKwData{NULL}{NULL}
\SetCommentSty{small}
$Q:=\emptyset$\tcp*[r]{$Q$ is the set of  jobs in the queue}
\While{the system is ON} {
\If{job $i$ arrives}{
$\xi_i:=k_i$\tcp*[r]{job $i$ has $\xi_i$ remaining tasks~~~~~~~}
$\gamma_i:=k_i$\tcp*[r]{job $i$ has $\gamma_i$ unassigned tasks~~~~~~}
$Q:=Q \cup \{i\}$\;
}

\While{$\sum_{i\in Q} \gamma_i> 0$ and there are idle servers}{
Pick any idle server $l$\;
$j :=\arg\min\{\gamma_i : i\in Q, \gamma_i>0\}$\;
Allocate a task of job $j$ on server $l$\;
$\gamma_j:=\gamma_j - 1$\;

}

\If{a task of job $i$ is completed}{

$\xi_{i} := \xi_{i} -1$\;
\textbf{if} $\xi_{i}=0$ \textbf{then} {$Q:=Q/ \{i\}$\;   }
}

}
\caption{ Fewest Unassigned Tasks first (FUT).}\label{alg1}
\end{algorithm}
\section{Main Results}\label{sec_analysis}

In this section, we provide some near delay optimality results in the form of \eqref{eq_nearoptimal} and \eqref{eq_nearoptimal1}. 
These results will be proven in Section \ref{sec_proofmain} by using a unified sample-path method.

\subsection{Average Delay $D_{\text{avg}}$ and the Delay Metrics in $\mathcal{D}_{\text{sym}}$}\label{sec_average_delay}

Let us first consider the following \textbf{Fewest Unassigned Tasks first (FUT)} policy, which is also described in Algorithm \ref{alg1}: 
Each idle server is allocated to process a task from the job with the earliest due time among all jobs with unassigned tasks. The delay performance of policy FUT is characterized in the following theorem.

\begin{theorem}\label{thm1}
If the task service times are NBU, independent across the servers, and {i.i.d.} across the tasks assigned to the same server, then for all \emph{$f\in\mathcal{D}_{\text{sym}}$},  ${\pi\in\Pi}$, and $\mathcal{I}$ 
\emph{
\begin{align}\label{eq_delaygap1}
&[f(\bm{V}(\text{FUT}))|\mathcal{I}] \leq_{\text{st}} \left[f(\bm{C}(\pi))|\mathcal{I}\right].
\end{align}}
\end{theorem}
If there are multiple idle servers, the choices of the idle server in Step 9 of Algorithm \ref{alg1} are non-unique. 
The result of Theorem \ref{thm1} holds for any server selection method in Step 9. In practice, one can pick the idle server with the fastest service rate.

Let us characterize the sub-optimality delay gap of policy FUT. For mean average delay, i.e., $f(\cdot)=D_{\text{avg}}(\cdot)$, it follows from \eqref{eq_delaygap1} that 
\begin{align}
\mathbb{E}[D_{\text{avg}}(\bm{V}(\text{FUT}))|\mathcal{I}] \leq 
\min_{\pi\in\Pi}\mathbb{E}\left[D_{\text{avg}}(\bm{C}(\pi))\left.\right|\mathcal{I}\right] \leq \mathbb{E}[D_{\text{avg}}(\bm{C}(\text{FUT}))|\mathcal{I}]. \label{eq_delaygap1_2}
\end{align}
The delay gap between the LHS and RHS of \eqref{eq_delaygap1_2} is given by 
\begin{align}\label{eq_delaygap1_1_1}
& \mathbb{E}[D_{\text{avg}}(\bm{C}(\text{FUT}))-D_{\text{avg}}(\bm{V}(\text{FUT}))|\mathcal{I}]\nonumber\\
=& \mathbb{E}\bigg[\frac{1}{n}\sum_{i=1}^n\big({C}_i(\text{FUT})-V_i(\text{FUT})\big)\bigg|\mathcal{I}\bigg].\!
\end{align}
Recall that $m$ is the number of servers, and $k_i$ is the number of tasks in job $i$. At time $V_i(\text{FUT})$, all tasks of job $i$ have started service: If $k_i>m$,  then job $i$ has at most $m$ incomplete tasks that are under service at time $V_i(\text{FUT})$; if $k_i\leq m$, then job $i$ has at most $k_i$ incomplete tasks that are under service at time $V_i(\text{FUT})$. Therefore, in policy FUT, at most $k_i\wedge m = \min\{k_i,m\}$ tasks of job $i$ are under service during the time interval $[V_i(\text{FUT}), {C}_i(\text{FUT})]$. Using this and the property of NBU distributions, we can obtain

\begin{theorem}\label{lem7_NBU}
Let $\mathbb{E}[X_l]= 1/\mu_l$, and without loss of generality $\mu_1\leq \mu_2\leq \ldots\leq \mu_m$. 
If the task service times are NBU, independent across the servers, {i.i.d.} across the tasks assigned to the same server, then  for all $\mathcal{I}$
\emph{\begin{align}\label{eq_gap}
\mathbb{E}[D_{\text{avg}}(\bm{C}(\text{FUT}))|\mathcal{I}] - \min_{\pi\in\Pi}\mathbb{E}\left[D_{\text{avg}}(\bm{C}(\pi))|\mathcal{I}\right]  \leq&\frac{1}{n}\sum_{i=1}^n\sum_{l=1}^{k_i \wedge m } \frac{1}{\sum_{j=1}^l\mu_j}\\
\leq& \frac{\ln(k_{\max}\wedge m)+1}{\mu_1},\label{eq_gap_2}
\end{align}}
where $k_{\max}$ is the maximum job size in \eqref{eq_0}, and $x\wedge y = \min\{x,y\}$.
\end{theorem}

If $\mathbb{E}[X_l]\leq 1/\mu$ for $l=1,\ldots,m$, then  the sub-optimality delay gap of policy FUT is of the order $O(\ln(k_{\max}\wedge m))/{\mu}$. As the number of servers $m$ increases, this  sub-optimality delay gap is upper bounded by $[\ln(k_{\max})+1]/{\mu}$ which is independent of $m$. 

Theorem \ref{thm1} and Theorem \ref{lem7_NBU} tell us that {for arbitrary number, batch sizes, arrival times, and due times of the jobs, the FUT policy is near delay-optimal for minimizing the average delay $D_{\text{avg}}$ in the mean. 

The FUT policy is a nice approximation of the Shortest Remaining Processing Time (SRPT) first policy  \cite{Schrage68,Smith78}: The FUT policy  utilizes the number of unassigned tasks of a job to approximate the remaining processing time of this job, and is within a small additive sub-optimality gap from the optimum for minimizing the mean average delay in the scheduling problem that we consider.

Due to the difficulty of delay minimization in multi-class multi-server queueing systems, attempts on finding a tight additive sub-optimal delay gap in these systems has met with little success. In \cite{Weiss:1992,Weiss:1995}, closed-form upper bounds on the sub-optimality gaps of the Smith's rule and the Gittin's index rule were established for the cases that all jobs arrive at time zero. In Corollary 2 of \cite{Dacre1999}, the authors studied the optimal control in multi-server queueing systems with stationary arrivals of multiple classes of jobs, and used the achievable region method to obtain an additive sub-optimality delay gap, which is of the order $O(m k_{\max})$. The model and methodology in these studies are significantly different from those in this paper.

\subsection{Maximum Lateness $L_{\max}$ and the Delay Metrics in $\mathcal{D}_{\text{sym}}\cup\mathcal{D}_{\text{Sch-1}}$}
Next, we investigate the maximum lateness $L_{\max}$ and the delay metrics in $\mathcal{D}_{\text{sym}}\cup\mathcal{D}_{\text{Sch-1}}$. We consider a   \textbf{Earliest Due Date first (EDD)} policy: Each idle server is allocated to process a task from the job with the earliest due time among all jobs with unassigned tasks. Policy EDD can be obtained from Algorithm \ref{alg1} by revising step 10 as $j :=\arg\min\{d_i : i\in Q, \gamma_i>0\}$.
The delay performance of policy EDD is characterized in the following theorem.
\begin{theorem}\label{thm3}
If the task service times are NBU, independent across the servers, and {i.i.d.} across the tasks assigned to the same server, then for all ${\pi\in\Pi}$, and $\mathcal{I}$ 
\emph{
\begin{align}\label{eq_delaygap1_thm3}
&[L_{\max}(\bm{V}(\text{EDD}))|\mathcal{I}] \leq_{\text{st}} \left[L_{\max}(\bm{C}(\pi))\left.\right|\mathcal{I}\right].
\end{align}}
\end{theorem}

If the job sizes satisfy certain conditions, policy EDD is identical with policy FUT. In this case, policy EDD is near delay-optimal in stochastic ordering for minimizing both the average delay $D_{\text{avg}}(\cdot)$ and the maximum lateness $L_{\max}(\cdot)$. Interestingly, in this case policy EDD is near delay-optimal in stochastic ordering for minimizing all delay metrics in {$\mathcal{D}_{\text{sym}}\cup\mathcal{D}_{\text{Sch-1}}$}, as stated in the following theorem:



\begin{theorem}\label{coro_thm3_1}
If $k_1=\ldots=k_n=1$ (or $d_1\leq d_2\leq \ldots\leq d_n$ and $k_1\leq k_2 \leq \ldots\leq k_n$) and \emph{$L_{\max}$} is replaced by any \emph{$f\in\mathcal{D}_{\text{sym}}\cup\mathcal{D}_{\text{Sch-1}}$}, Theorem  \ref{thm3} still holds.
\end{theorem}

\subsection{Maximum Delay $D_{\max}$ and the Delay Metrics in $\mathcal{D}_{\text{sym}}\cup\mathcal{D}_{\text{Sch-2}}$}
We now study maximum delay $D_{\max}$ and the delay metrics in $\mathcal{D}_{\text{sym}}\cup\mathcal{D}_{\text{Sch-2}}$. We consider a \textbf{First-Come, First-Served  (FCFS)} policy: Each idle server is allocated to process a task from the job with the earliest arrival time among all jobs with unassigned tasks. Policy FCFS can be obtained from Algorithm \ref{alg1} by revising step 10 as $j :=\arg\min\{a_i : i\in Q, \gamma_i>0\}$. The delay performance of FCFS is characterized as follows:

\begin{corollary}\label{thm5}
If the task service times are NBU, independent across the servers, and {i.i.d.} across the tasks assigned to the same server, then for all  $\pi\in\Pi$, and $\mathcal{I}$
\emph{
\begin{align}\label{eq_delaygap1_thm5}
&[D_{\max}(\bm{V}(\text{FCFS}))|\mathcal{I}] \leq_{\text{st}} \left[D_{\max}(\bm{C}(\pi))|\mathcal{I}\right].
\end{align}}
\end{corollary}

In addition, if the job sizes satisfy certain conditions, policy FCFS is identical with policy FUT. Hence, policy FCFS is near delay-optimal in stochastic ordering for minimizing both the average delay $D_{\text{avg}}(\cdot)$ and the maximum delay $D_{\max}(\cdot)$. In this case, we can further show that policy FCFS is  near delay-optimal in distribution for minimizing all delay metrics in {$\mathcal{D}_{\text{sym}}\cup\mathcal{D}_{\text{Sch-2}}$}, as stated in the following:

\begin{corollary}\label{coro5_1}
If $k_1\leq k_2 \leq \ldots\leq k_n$ and \emph{$D_{\max}$} is replaced by any \emph{$f\in\mathcal{D}_{\text{sym}}\cup\mathcal{D}_{\text{Sch-2}}$}, Theorem  \ref{thm5} still holds.
\end{corollary}
The proofs of Corollaries \ref{thm5} and \ref{coro5_1} are omitted, because they follow directly from Theorems \ref{thm3} and \ref{coro_thm3_1} by
setting $d_i = a_i$ for all job $i$. 

In the sequel, we show that policy FCFS is within twice of the optimum for minimizing the mean maximum delay $\mathbb{E}\left[D_{\max}(\bm{C}(\pi))|\mathcal{I}\right]$ and mean $p$-norm of delay $\mathbb{E}\left[|| \bm{D}(\pi) ||_p|\mathcal{I}\right]$ for $p\geq 1$.\footnote{Because of the maximum operator over all jobs, $\mathbb{E}\left[D_{\max}(\bm{C}(\pi))|\mathcal{I}\right]$ will likely grow to infinite as the number of jobs $n$ increases. Hence, it is difficult to obtain an additive sub-optimality gap for minimizing $\mathbb{E}\left[D_{\max}(\bm{C}(\pi))|\mathcal{I}\right]$ that remains constant for all $n$.} 

\begin{lemma}\label{lem_ratio}
For arbitrarily given job parameters $\mathcal{I}$, if (i) the task service times are NBU and {i.i.d.} across the tasks and  servers, (ii) $f(\bm{x} + \bm{a})$ is sub-additive and increasing in $\bm{x}$, and (iii) two policies $P,\pi\in \Pi$ satisfy
\emph{ 
\begin{align}\label{eq_ratio_0}
 [f(\bm{V}(P))|\mathcal{I}] \leq_{\text{st}} \left[f(\bm{C}(\pi))|\mathcal{I}\right],
\end{align}}
then 
\emph{
\begin{align}\label{eq_ratio_1}
 [f(\bm{C}(P))|\mathcal{I}] \leq_{\text{st}}  \left[ 2 f(\bm{C}(\pi))|\mathcal{I}\right].
\end{align}}
\end{lemma}
\begin{proof}
See Appendix \ref{app_lem_ratio}.
\end{proof}

By applying Lemma \ref{lem_ratio} in Corollaries \ref{thm5} and \ref{coro5_1}, it is easy to obtain
\begin{theorem}\label{thm_max_delay_ratio}
If the task service times are NBU and {i.i.d.} across the tasks and  servers, then for all $\mathcal{I}$
\emph{
\begin{align}\label{}
&[D_{\max}(\bm{C}(\text{FCFS}))|\mathcal{I}] \leq_{\text{st}} [2 D_{\max}(\bm{C}(\pi))|\mathcal{I}],~\forall~\pi\in\Pi. \nonumber\\
&\mathbb{E}[D_{\max}(\bm{C}(\text{FCFS}))|\mathcal{I}] \leq\min_{\pi\in\Pi} 2 \mathbb{E}\left[D_{\max}(\bm{C}(\pi))|\mathcal{I}\right].\nonumber
\end{align}}
\end{theorem}
\begin{theorem}\label{thm_max_delay_ratio_1}
If (i) $k_1\leq k_2 \leq \ldots\leq k_n$, (ii) the task service times are NBU and {i.i.d.} across the tasks and  servers, (iii) $f(\bm{x} + \bm{a})$ is sub-additive and increasing in $\bm{x}$, (iv) \emph{$f\in\mathcal{D}_{\text{sym}}\cup\mathcal{D}_{\text{Sch-2}}$}, then for all $\mathcal{I}$
\emph{
\begin{align}
&[f(\bm{C}(\text{FCFS}))|\mathcal{I}] \leq_{\text{st}} [2 f(\bm{C}(\pi))|\mathcal{I}],~\forall~\pi\in\Pi. \nonumber\\
&\mathbb{E}[f(\bm{C}(\text{FCFS}))|\mathcal{I}] \leq\min_{\pi\in\Pi} 2 \mathbb{E}\left[f(\bm{C}(\pi))|\mathcal{I}\right].\nonumber
\end{align}}
\end{theorem}
There exist a fair large class of delay metrics satisfying the conditions (iii) and (iv) of Theorem \ref{thm_max_delay_ratio_1}. For example, consider the {$p$-norm of delay} $D_p(\bm{C}(\pi))=D_p(\bm{D}(\pi)+\bm{a}) = || \bm{D}(\pi) ||_p$ for $p\geq 1$. 
It is known that any norm is sub-additive, i.e., $|| \bm{x} + \bm{y} || \leq || \bm{x} || + || \bm{y} ||$. In addition, $|| \bm{D}(\pi) ||_p$ is symmetric and convex  in the delay vector $\bm{D}(\pi)$, and hence is Schur-convex in $\bm{D}(\pi)$. Therefore,
 $f = {D}_{p}$ satisfies the conditions (iii) and (iv) of Theorem \ref{thm_max_delay_ratio_1} for all $p\geq 1$.




\section{Proofs of the Main Results} \label{sec_proofmain}
We propose a unified sample-path method to prove the main results in the previous section. This method contains three steps: 

\begin{itemize}
\item[\emph{1.}] \textbf{Sample-Path Orderings:} We first introduce several novel sample-path orderings (Propositions \ref{ordering_2}, \ref{ordering_3_1}, and Corollary \ref{ordering_4_1}).
Each of these sample-path orderings can be used to obtain a delay inequality for comparing the delay performance (i.e., average delay, maximum lateness, or maximum delay) of different policies in a near-optimal sense.
\item[\emph{2.}]  \textbf{Sufficient Conditions for Sample-Path Orderings:} In order to minimize delay, it is important to execute the tasks as fast as possible. Motivated by this, we introduce a weak work-efficiency ordering to compare the efficiency of task executions in different policies. By combining this weak work-efficiency ordering with appropriate priority rules (i.e., FUT, EDD, FCFS) for job services, we obtain several sufficient conditions  (Propositions \ref{lem1}-\ref{lem2}) of the sample-path orderings in \emph{Step 1}. In addition, if more than one sufficient conditions are satisfied simultaneously, we are able to obtain delay inequalities for comparing more general classes of delay metrics achieved by different policies (Propositions \ref{lem1_1}-\ref{lem_general}).

\item[\emph{3.}]  \textbf{Coupling Arguments:} We use coupling arguments to prove that for NBU task service times, the weak work-efficiency ordering is satisfied in the sense of stochastic ordering. By combining this with the priority rules (i.e., FUT, EDD, FCFS) for job services, we are able to prove the sufficient conditions in \emph{Step 2} in the sense of stochastic ordering. By this, the main results of this paper are proven. 
\end{itemize}
This sample-path method is quite general. In particular, \emph{Step 1} and \emph{Step 2} do not need to specify the queueing system model, and can be potentially used for establishing near delay optimality results in other queueing systems. One application of this sample-path method is in \cite{Yin_report2016}, which considered the problem of scheduling in queueing systems with replications.

\subsection{Step 1: Sample-path Orderings}
\label{sec_sufficient}

We first propose several sample-path orderings to compare the delay performance of different scheduling policies. Let us first define the system state of any policy $\pi\in\Pi$. 

\begin{definition}\label{def_state_thm3}
At any time instant $t\in[0,\infty)$, the system state of policy $\pi$ is specified by a pair of $n$-dimensional vectors $\bm{\xi}_{\pi}(t)=(\xi_{1,\pi}(t),\ldots,\xi_{n,\pi}(t))$ and $\bm{\gamma}_{\pi}(t)=(\gamma_{1,\pi}(t),\ldots,\gamma_{n,\pi}(t))$ with non-negative components, where $n$ is the total number of jobs and can be either finite or infinite. The components of $\bm{\xi}_{\pi}(t)$ and $\bm{\gamma}_{\pi}(t)$ are interpreted as follows:
If job $i$ is present in the system at time $t$, then $\xi_{i,\pi}(t)$ is the number of \emph{remaining} tasks (which are either stored in the queue or being executed by the servers) of job $i$, and $\gamma_{i,\pi}(t)$ is the number of \emph{unassigned} tasks (which are stored in the queue and not being executed by the servers) of job $i$; if job $i$ is not present in the system at time $t$ (i.e., job $i$ has not arrived at the system or has departed from the system), then $\xi_{i,\pi}(t)=\gamma_{i,\pi}(t)=0$. Hence, for all $i=1,\ldots,n$, $\pi\in\Pi$, and $t\in[0,\infty)$
\begin{align}\label{def_relation}
0\leq \gamma_{i,\pi}(t)\leq \xi_{i,\pi}(t)\leq k_i.
\end{align}

Let $\{\bm{\xi}_{\pi}(t),\bm{\gamma}_{\pi}(t),t\in[0,\infty)\}$ denote the state process of policy $\pi$ in a probability space $(\Omega,\mathcal{F},P)$, which is assumed to be right-continuous. The realization of the state process on a sample path $\omega\in \Omega$ can be expressed as $\{\bm{\xi}_{\pi}(\omega,t),$ $\bm{\gamma}_{\pi}(\omega,t),t\in[0,\infty)\}$. To ease the notational burden, we will omit $\omega$ henceforth and reuse $\{\bm{\xi}_{\pi}(t),\bm{\gamma}_{\pi}(t),t\in[0,\infty)\}$ to denote the realization of the state process on a sample path. Because the system starts to operate at time $t=0$,  there is no job in the system before time $t=0$. Hence, 
$\bm{\xi}_{\pi}(0^-) = \bm{\gamma}_{\pi}(0^-)=\bm{0}$. 
\end{definition}

%





The following proposition provides one condition \eqref{eq_ordering_1_1} for comparing the average delay $D_{\text{avg}}(\bm{c}(\pi))$ of different policies on a sample path, which was firstly used in \cite{Smith78} to prove the optimality of the preemptive SRPT policy to minimize the average delay in single-server queueing systems. 
 \begin{proposition} \label{ordering_1} 
For any given job parameters $\mathcal{I}$ and a sample path of two policies $P,\pi\in\Pi$, 
if 
\begin{eqnarray}\label{eq_ordering_1_1}
\sum_{i=j}^n {\xi}_{[i],P}(t)\leq \sum_{i=j}^n {\xi}_{[i],\pi}(t),~\forall~j = 1,2,\ldots,n,
\end{eqnarray}
holds for all $t\in[0,\infty)$,\footnote{In majorization theory \cite{Marshall2011},  \eqref{eq_ordering_1_1} is equivalent to ``$\bm{\xi}_{\pi}(t)$ is weakly supermajorized by $\bm{\xi}_{P}(t)$, i.e., $\bm{\xi}_{\pi}(t)\prec^{\text{w}}\bm{\xi}_{P}(t)$''.} where ${\xi}_{[i],\pi}(t)$ is the $i$-th largest component of $\bm{\xi}_{\pi}(t)$, then  
\begin{align}\label{eq_ordering_1_1_2}
c_{(i)}(P)\leq c_{(i)}(\pi),~\forall~i=1,2,\ldots,n,
\end{align}
where $c_{(i)}(\pi)$ is the $i$-th smallest component of $\bm{c}(\pi)$.\footnote{In other words, $c_{(i)}(\pi)$ is the earliest time in policy $\pi$ by which $i$ jobs have been completed.} Hence,
\emph{
\begin{align}\label{eq_ordering_1_2}
D_{\text{avg}}(\bm{c}(P))\leq D_{\text{avg}}(\bm{c}(\pi)).
\end{align}
}
\end{proposition}

\begin{proof}
Suppose that there are $l$ unfinished jobs at time $t$ in policy $\pi$, then $\sum_{i=l+1}^n {\xi}_{[i],\pi}(t)=0$. By \eqref{eq_ordering_1_1}, we get $\sum_{i=l+1}^n {\xi}_{[i],P}(t)=0$ and hence there are at most $l$ unfinished jobs in policy $P$. In other words, there are at least as many unfinished jobs in policy $\pi$ as in policy $P$ at any time $t\in[0,\infty)$. This implies \eqref{eq_ordering_1_1_2}, because the sequence of job arrival times $a_1,a_2,\ldots, a_n$ are invariant under any policy. In addition, \eqref{eq_ordering_1_2} follows from \eqref{eq_ordering_1_1_2}, which completes the proof.
\end{proof}

The sample-path ordering  \eqref{eq_ordering_1_1} is quite insightful. According to Proposition \ref{ordering_1}, if \eqref{eq_ordering_1_1} holds for all policies $\pi\in\Pi$ and all sample paths $\omega\in \Omega$, then policy $P$ is sample-path delay-optimal for minimizing the average delay $D_{\text{avg}}(\bm{c}(\pi))$. 
Interestingly, Proposition \ref{ordering_1} is also necessary: If \eqref{eq_ordering_1_1} does not hold at some time $t$, then one can construct an arrival process after time $t$ such that \eqref{eq_ordering_1_1_2} and \eqref{eq_ordering_1_2} do not hold  \cite{Smith78}. 

The sample-path ordering  \eqref{eq_ordering_1_1} has been successfully used to analyze single-server queueing systems \cite{Smith78}. However, applying this condition in multi-server queueing systems is challenging. 
Due to this fundamental difficulty, we consider an alternative method to relax the sample-path ordering \eqref{eq_ordering_1_1} and seek for near delay optimality.

 \begin{proposition} \label{ordering_2} 
For any given job parameters $\mathcal{I}$ and a sample path of two policies $P,\pi\in\Pi$, if 
\begin{eqnarray}\label{eq_ordering_2_1}
\sum_{i=j}^n {\gamma}_{[i],P}(t)\leq \sum_{i=j}^n {\xi}_{[i],\pi}(t),~\forall~j = 1,2,\ldots,n,
\end{eqnarray}
holds for all $t\in[0,\infty)$, where ${\gamma}_{[i],\pi}(t)$ is the $i$-th largest component of $\bm{\gamma}_{\pi}(t)$, then 
\begin{align}\label{eq_ordering_2_2}
v_{(i)}(P)\leq c_{(i)}(\pi),~\forall~i=1,2,\ldots,n,
\end{align}
where $v_{(i)}(P)$ is the $i$-th smallest component of $\bm{v}(P)$.\footnote{In other words, $v_{(i)}(P)$ is the earliest time in policy $P$ that there exist $i$ jobs whose tasks have all started service.}
Hence, 
\emph{
\begin{align}\label{eq_ordering_2_3}
D_{\text{avg}} (\bm{v}(P))\leq D_{\text{avg}}(\bm{c}(\pi)).
\end{align}}
\end{proposition}


\begin{proof}
See Appendix \ref{app_lem2}. 
\end{proof}

Hence, by relaxing the sample-path ordering \eqref{eq_ordering_1_1} as \eqref{eq_ordering_2_1}, a relaxed delay inequality \eqref{eq_ordering_2_3} is obtained which can be used to compare the average delay of policy $P$ and policy $\pi$ in a near-optimal sense. 



Similarly, we develop two sample-path orderings in the following two lemmas to compare the maximum lateness $L_{\max}(\cdot)$ achieved by different policies.




 \begin{proposition} \label{ordering_3} 
For any given job parameters $\mathcal{I}$ and a sample path of two policies $P,\pi\in\Pi$, if 
\begin{eqnarray}\label{eq_ordering_3_1}
\sum_{i: d_i \leq \tau} {\xi}_{i,P}(t)\leq \sum_{i: d_i \leq \tau} {\xi}_{i,\pi}(t),~\forall~\tau\in[0,\infty),
\end{eqnarray}
holds for all $t\in[0,\infty)$, then 
\emph{
\begin{align}\label{eq_ordering_3_2}
L_{\max} (\bm{c}(P))\leq L_{\max}(\bm{c}(\pi)).
\end{align}}
\end{proposition}
\begin{proof}
See Appendix \ref{app_lem3}. 
\end{proof}
\begin{proposition} \label{ordering_3_1} 
For any given job parameters $\mathcal{I}$ and a sample path of two policies $P,\pi\in\Pi$, if 
\begin{eqnarray}\label{eq_ordering_3_3}
\sum_{i: d_i \leq \tau} {\gamma}_{i,P}(t)\leq \sum_{i: d_i \leq \tau} {\xi}_{i,\pi}(t),~\forall~\tau\in[0,\infty),
\end{eqnarray}
holds for all $t\in[0,\infty)$, then 
\emph{
\begin{align}\label{eq_ordering_3_4}
L_{\max} (\bm{v}(P))\leq L_{\max}(\bm{c}(\pi)).
\end{align}} 
\end{proposition}

\begin{proof}
See Appendix \ref{app_lem3}. 
\end{proof}

%
%
%

If $d_i=a_i$ for all $i$,  the maximum lateness $L_{\max}(\cdot)$ reduces to the {maximum delay} $D_{\max}(\cdot)$. Hence, we can obtain 

 \begin{corollary} \label{ordering_4} 
For any given job parameters $\mathcal{I}$ and a sample path of two policies $P,\pi\in\Pi$, if 
\begin{eqnarray}\label{eq_ordering_4_1}
\sum_{i: a_i \leq \tau} {\xi}_{i,P}(t)\leq \sum_{i: a_i \leq \tau} {\xi}_{i,\pi}(t),~\forall~\tau\in[0,\infty),
\end{eqnarray}
holds for all $t\in[0,\infty)$, then 
\emph{
\begin{align}\label{eq_ordering_4_2}
D_{\max} (\bm{c}(P))\leq D_{\max}(\bm{c}(\pi)).
\end{align}}
\end{corollary}
\begin{corollary} \label{ordering_4_1} 
For any given job parameters $\mathcal{I}$ and a sample path of two policies $P,\pi\in\Pi$, if
\begin{eqnarray}\label{eq_ordering_4_3}
\sum_{i: a_i \leq \tau} {\gamma}_{i,P}(t)\leq \sum_{i: a_i \leq \tau} {\xi}_{i,\pi}(t),~\forall~\tau\in[0,\infty),
\end{eqnarray}
holds for all $t\in[0,\infty)$, then 
\emph{
\begin{align}\label{eq_ordering_4_4}
D_{\max} (\bm{v}(P))\leq D_{\max}(\bm{c}(\pi)).
\end{align}}
\end{corollary}

The proofs of Corollaries \ref{ordering_4}-\ref{ordering_4_1} are omitted, because they follow directly from Propositions \ref{ordering_3}-\ref{ordering_3_1} by setting $d_i=a_i$ for all $i$. 

The sample-path orderings in Propositions \ref{ordering_1}-\ref{ordering_3_1} and Corollaries \ref{ordering_4}-\ref{ordering_4_1} are of similar forms. Their distinct features are 
\begin{itemize} 
\item In the sample-path orderings \eqref{eq_ordering_1_1} and \eqref{eq_ordering_2_1} corresponding to the average delay $D_{\text{avg}}(\cdot)$, the summations are taken over the jobs with the fewest remaining/unassigned tasks;
 
\item In the sample-path orderings \eqref{eq_ordering_3_1} and \eqref{eq_ordering_3_3} corresponding to the maximum lateness $L_{\max}(\cdot)$, the summations are taken over the jobs with the earliest due times;
 
\item In the sample-path orderings \eqref{eq_ordering_4_1} and \eqref{eq_ordering_4_3} corresponding to the maximum delay $D_{\max}(\cdot)$, the summations are taken over the jobs with the earliest arrival times.
\end{itemize} 
These features are tightly related to the priority rules for minimizing the corresponding delay metrics near optimally: The priority rule for minimizing the average delay $D_{\text{avg}}(\cdot)$ is 
FUT first; the priority rule for minimizing the maximum lateness $L_{\max}(\cdot)$ is EDD first; the priority rule for minimizing the maximum delay $D_{\max}(\cdot)$ is FCFS. Hence, \emph{the summations in these sample-path orderings are taken over the high priority jobs}. This is one key insight behind these sample-path orderings.

A number of popular sample-path methods --- such as forward induction, backward induction, and interchange arguments  \cite{Liu1995} --- have been successfully used to establish delay optimality results in single-server queueing systems \cite{Schrage68,Jackson55,Baccelli:1993}. However, it appears to be difficult to directly generalize these methods and characterize the sub-optimal delay gap from the optimum. 
On the other hand, the sample-path orderings \eqref{eq_ordering_1_1}, \eqref{eq_ordering_2_1}, \eqref{eq_ordering_3_1}, \eqref{eq_ordering_3_3}, \eqref{eq_ordering_4_1}, and \eqref{eq_ordering_4_3} provide an interesting unified framework for  sample-path delay comparisons, aiming for both delay optimality and near delay optimality. 
 To the best of our knowledge, except for \eqref{eq_ordering_1_1} developed in \cite{Smith78}, the sample-path orderings \eqref{eq_ordering_2_1}, \eqref{eq_ordering_3_1}, \eqref{eq_ordering_3_3}, \eqref{eq_ordering_4_1}, and \eqref{eq_ordering_4_3} have not appeared before.


\subsection{Step 2: Sufficient Conditions for Sample-path Orderings}\label{sec_delay_ineq}
Next, we provide several sufficient conditions for the sample-path orderings \eqref{eq_ordering_2_1}, \eqref{eq_ordering_3_3}, and \eqref{eq_ordering_4_3}. In addition, we will also develop several sufficient conditions for comparing more general delay metrics in $\mathcal{D}_{\text{sym}}$, $\mathcal{D}_{\text{Sch-1}}$, and $\mathcal{D}_{\text{Sch-2}}$.

\subsubsection{Weak Work-efficiency Ordering}
In single-server queueing systems, the service delay is largely governed by the work conservation law (or its generalizations): At any time, the expected total amount of time for completing the jobs in the queue is invariant among all work-conserving policies \cite{Leonard_Kleinrock_book,Jose2010,Gittins:11}. However, this work conservation law does not hold in queueing systems with parallel servers, where it is difficult to efficiently utilize the full service capacity (some servers can be idle if all tasks are being processed by the remaining servers). In order to minimize delay, the system needs to execute the tasks as fast as possible.
In the sequel, we introduce an ordering to compare the efficiency of task executions in different policies in a near-optimal sense, which is called \emph{weak work-efficiency ordering}.\footnote{Work-efficiency orderings are also used in \cite{Yin_report2016} to study queueing systems with replications.}

\begin{definition} \label{def_order} \textbf{Weak Work-efficiency Ordering:}
For any given job parameters $\mathcal{I}$ and a sample path of two policies $P,\pi\in\Pi$, policy $P$ is said to be \textbf{weakly more work-efficient than} policy $\pi$, if the following assertion is true:
For each task $j$ executed in policy $\pi$, if
\begin{itemize}
\item[1.] In policy $\pi$, task $j$ starts service at time $\tau$ and completes service at time $\nu$ ($\tau\leq \nu$), 
\item[2.] In policy $P$, the queue is not empty (there exist unassigned tasks in the queue) during $[\tau,\nu]$, 
\end{itemize}
then there exists one corresponding task $j'$ in policy $P$ which starts service during $[\tau,\nu]$. \end{definition}

\begin{figure}
\centering 
\includegraphics[width=0.35\textwidth]{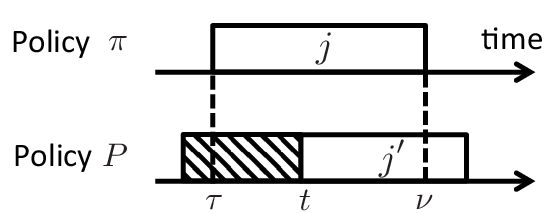} \caption{Illustration of the weak work-efficiency ordering, where the service duration of a task is indicated by a rectangle. 
Task $j$ starts service at time $\tau$ and completes service at time $\nu$ in policy $\pi$, and one corresponding task $j'$ starts service at  time $t\in[\tau,\nu]$ in policy $P$.}
\label{Work_Efficiency_Ordering} 
\end{figure} 

An illustration of this weak work-efficiency ordering is provided in Fig. \ref{Work_Efficiency_Ordering}. Note that the weak work-efficient ordering requires the service starting time of task $j'$ in policy $P$ to be within the service duration of its corresponding task $j$ in policy $\pi$. This is the key feature that enables establishing tight sub-optimality delay gap later on. 
\subsubsection{Sufficient Conditions for  Sample-path Orderings}
Using weak work-efficiency ordering, we can obtain 
the following sufficient condition for the sample-path ordering \eqref{eq_ordering_2_1} associated to the average delay $D_{\text{avg}}(\cdot)$.

\begin{proposition}\label{lem1}
For any given job parameters $\mathcal{I}$ and a sample path of two policies $P,\pi\in\Pi$, if
\begin{itemize}
\itemsep0em 
\item[1.] Policy $P$ is {weakly} {more work-efficient} than policy $\pi$,
\item[2.] In policy $P$, each task starting service is from the job with the fewest unassigned tasks among all jobs with unassigned tasks, 
\end{itemize} 
then \eqref{eq_ordering_2_1}-\eqref{eq_ordering_2_3} hold.  
\end{proposition}
\begin{proof}
See Appendix \ref{app0}.
\end{proof}



Similarly, one sufficient condition is obtained for the sample-path orderings \eqref{eq_ordering_3_3} for comparing the maximum lateness $L_{\max}(\cdot)$ of different policies.

\begin{proposition}\label{lem2}
For any given job parameters $\mathcal{I}$ and a sample path of two policies $P,\pi\in\Pi$, if
\begin{itemize}
\itemsep0em 
\item[1.] Policy $P$ is {weakly} {more work-efficient} than policy $\pi$,
\item[2.] In policy $P$, each task starting service is from the job with the earliest due time among all jobs with unassigned tasks, 
\end{itemize} 
then \eqref{eq_ordering_3_3} and \eqref{eq_ordering_3_4} hold. \end{proposition}

\begin{proof}
See Appendix \ref{app0_1}. 
\end{proof}

%

\subsubsection{More General Delay Metrics}

We now investigate more general delay metrics in $\mathcal{D}_{\text{sym}}$, $\mathcal{D}_{\text{Sch-1}}$, and $\mathcal{D}_{\text{Sch-2}}$. First, Proposition \ref{lem1} can be directly generalized to all delay metrics in $\mathcal{D}_{\text{sym}}$. 

\begin{proposition}\label{lem1_1}
If the conditions of Proposition \ref{lem1} are satisfied, then for all \emph{$f\in\mathcal{D}_{\text{sym}}$}
\begin{align}
f (\bm{v}(P))\leq f(\bm{c}(\pi)).\nonumber
\end{align}
\end{proposition}
\begin{proof}
See Appendix \ref{app_lem1_1}.
\end{proof}

In some scenarios, there exists a policy $P$ simultaneously satisfying the sufficient conditions in Proposition \ref{lem1} (or Proposition \ref{lem1_1}) and Proposition \ref{lem2}. In this case, we can obtain a delay inequality for comparing any delay metric in $\mathcal{D}_{\text{sym}}\cup\mathcal{D}_{\text{Sch-1}}$.

\begin{proposition}\label{lem_general}
For any given job parameters $\mathcal{I}$ and a sample path of two policies $P,\pi\in\Pi$, if
\begin{itemize}
\itemsep0em 
\item[1.] Policy $P$ is {weakly} {more work-efficient} than policy $\pi$,
\item[2.] In policy $P$, each task starting service is from the job with the fewest unassigned tasks among all jobs with unassigned tasks, and this job is also the job with the earliest due time among all jobs with unassigned tasks, 
\end{itemize} 
then for all \emph{$f\in\mathcal{D}_{\text{sym}}\cup\mathcal{D}_{\text{Sch-1}}$}
\begin{align}\label{eq_lem_general}
f (\bm{v}(P))\leq f(\bm{c}(\pi)).
\end{align}
\end{proposition}

\begin{proof}[Proof sketch of Proposition \ref{lem_general}]
For any $f\in\mathcal{D}_{\text{sym}}$, \eqref{eq_ordering_2_2} and \eqref{eq_lem_general} follow from Propositions \ref{lem1} and \ref{lem1_1}.
For any $f\in\mathcal{D}_{\text{Sch-1}}$, we construct an $n$-dimensional vector $\bm{v}'$ and show that  
\begin{align}\label{eq_coro2_6_0}
\bm{v}(P)-\bm{d}\prec \bm{v}'- \bm{d} \leq \bm{c} (\pi)-\bm{d},
\end{align}
where the first majorization ordering in \eqref{eq_coro2_6_0} follows from the rearrangement inequality \cite[Theorem 6.F.14]{Marshall2011},\cite{Chang93rearrangement_majorization}, and the second inequality in  \eqref{eq_coro2_6_0} is proven by using  \eqref{eq_ordering_2_2}.
This further implies 
\begin{align}\label{eq_coro2_3_0}
\bm{v}(P)-\bm{d}\prec_{\text{w}} \bm{c} (\pi)-\bm{d}. 
\end{align} 
Using this, we can show that \eqref{eq_lem_general} holds for all $f\in\mathcal{D}_{\text{Sch-1}}$.
The details are provided in Appendix \ref{app_coro2}.
\end{proof}

Using this proposition, we can obtain 

\begin{corollary}\label{coro2}
For any given job parameters $\mathcal{I}$ and a sample path of two policies $P,\pi\in\Pi$, if
\begin{itemize}
\itemsep0em 
\item[1.] $k_1=\ldots=k_n=1$ (or $d_1\leq d_2\leq \ldots\leq d_n$ and $k_1\leq k_2 \leq \ldots\leq k_n$),
\item[2.] Policy $P$ is {weakly} {more work-efficient} than policy $\pi$,
\item[3.] In policy $P$, each task starting service is from the job with the earliest due time among all jobs with unassigned tasks, 
\end{itemize} 
then 
for all \emph{$f\in\mathcal{D}_{\text{sym}}\cup\mathcal{D}_{\text{Sch-1}}$}
\begin{align}\label{eq_coro2}
f (\bm{v}(P))\leq f(\bm{c}(\pi)).
\end{align}
\end{corollary}
\begin{proof}
If $k_1=\ldots=k_n=1$, each job has only one task. Hence, the job with the earliest due time is also one job with the fewest unassigned tasks. If $d_1\leq d_2\leq \ldots\leq d_n$, $k_1\leq k_2 \leq \ldots\leq k_n$, and each task starting service is from the job with the earliest due time among all jobs with unassigned tasks, then each task starting service is also from the job with the fewest unassigned tasks among all jobs with unassigned tasks. 
By using 
Proposition \ref{lem_general}, Corollary \ref{coro2} follows. This completes the proof.
\end{proof}

\subsection{Step 3: Coupling Arguments}\label{sec_delay_ineq}

The following coupling lemma plays one important role in the proof of our main results.

\begin{lemma}\label{lem_coupling}
Consider two policies $P,\pi\in \Pi$. If policy $P$ is work-conserving, the task service times are NBU, independent across the servers, and {i.i.d.} across the tasks assigned to the same server, then there exist two state processes $\{\bm{\xi}_{P_1}(t),\bm{\gamma}_{P_1}(t),t\in[0,\infty)\}$ and $\{\bm{\xi}_{\pi_1}(t),\bm{\gamma}_{\pi_1}(t),t\in[0,\infty)\}$ of policy $P_1$ and  policy $\pi_1$, such that 
\begin{itemize}
\itemsep0em 
\item[1.] The state process $\{\bm{\xi}_{P_1}(t),\bm{\gamma}_{P_1}(t),t\in[0,\infty)\}$ of policy $P_1$ has the same distribution as the state process $\{\bm{\xi}_{P}(t),\bm{\gamma}_{P}(t),t\in[0,\infty)\}$ of policy $P$,
\item[2.] The state process $\{\bm{\xi}_{\pi_1}(t),\bm{\gamma}_{\pi_1}(t),t\in[0,\infty)\}$ of policy $\pi_1$ has the same distribution as the state process $\{\bm{\xi}_{\pi}(t),\bm{\gamma}_{\pi}(t),t\in[0,\infty)\}$  of policy $\pi$,
\item[3.] Policy $P_1$ is weakly more work-efficient than policy $\pi_1$ with probability one. 
\end{itemize} 
\end{lemma}

\begin{proof}
See Appendix \ref{app1}.
\end{proof}

Now, we are ready to prove our main results. 

\begin{proof}[Proof of Theorem \ref{thm1}]
 According to lemma \ref{lem_coupling}, for any policy $\pi\in\Pi$, there exist two state processes $\{\bm{\xi}_{\text{FUT}_1}(t),\bm{\gamma}_{\text{FUT}_1}(t),t\in[0,\infty)\}$ and $\{\bm{\xi}_{\pi_1}(t),\bm{\gamma}_{\pi_1}(t),t\in[0,\infty)\}$ of policy FUT$_1$ and  policy $\pi_1$, such that (i) the state process $\{\bm{\xi}_{\text{FUT}_1}(t),\bm{\gamma}_{\text{FUT}_1}(t),t\in[0,\infty)\}$ of policy FUT$_1$  has the same distribution as the state process $\{\bm{\xi}_{\text{FUT}}(t),\bm{\gamma}_{\text{FUT}}(t),t\in[0,\infty)\}$ of policy FUT, (ii) the state process  $\{\bm{\xi}_{\pi_1}(t),\bm{\gamma}_{\pi_1}(t),t\in[0,\infty)\}$ of policy   $\pi_1$ has the same distribution as the state process $\{\bm{\xi}_{\pi}(t),\bm{\gamma}_{\pi}(t),t\in[0,\infty)\}$ of policy $\pi$, and (iii) policy FUT$_1$ is weakly more work-efficient than policy $\pi_1$ with probability one.
By (iii) and Proposition \ref{lem1_1}, for all $f\in\mathcal{D}_{\text{sym}}$
\begin{align}
\Pr[f(\bm{V}(\text{FUT}_1)) \leq f(\bm{C}(\pi_1)) |\mathcal{I}]=1. \nonumber
\end{align}
By (i), $f(\bm{V}(\text{FUT}_1))$ has the same distribution as $f(\bm{V}(\text{FUT}))$. By (ii), $f(\bm{C}(\pi_1))$ has the same distribution as $f(\bm{C}(\pi))$.
Using the property of stochastic ordering \cite[Theorem 1.A.1]{StochasticOrderBook}, we can obtain \eqref{eq_delaygap1}. This completes the proof.
\end{proof}

\begin{proof}[Proof of Theorem \ref{lem7_NBU}]
After Theorem \ref{thm1} is established, Theorem \ref{lem7_NBU} is proven in Appendix \ref{app_lem7}.
\end{proof}

\begin{proof}[Proof of Theorem \ref{thm3}]
 According to lemma \ref{lem_coupling}, for any policy $\pi\in\Pi$, there exist two state processes $\{\bm{\xi}_{\text{EDD}_1}(t),\bm{\gamma}_{\text{EDD}_1}(t),t\in[0,\infty)\}$ and $\{\bm{\xi}_{\pi_1}(t),\bm{\gamma}_{\pi_1}(t),t\in[0,\infty)\}$ of policy EDD$_1$ and  policy $\pi_1$, such that (i) the state process $\{\bm{\xi}_{\text{EDD}_1}(t),\bm{\gamma}_{\text{EDD}_1}(t),t\in[0,\infty)\}$ of policy EDD$_1$  has the same distribution as the state process $\{\bm{\xi}_{\text{EDD}}(t),\bm{\gamma}_{\text{EDD}}(t),t\in[0,\infty)\}$ of policy EDD, (ii) the state process  $\{\bm{\xi}_{\pi_1}(t),\bm{\gamma}_{\pi_1}(t),t\in[0,\infty)\}$ of policy   $\pi_1$ has the same distribution as the state process $\{\bm{\xi}_{\pi}(t),\bm{\gamma}_{\pi}(t),t\in[0,\infty)\}$ of policy $\pi$, and (iii) policy EDD$_1$ is weakly more work-efficient than policy $\pi_1$ with probability one.
By (iii) and Proposition \ref{lem2}, 
\begin{align}
\Pr[L_{\max}(\bm{V}(\text{EDD}_1)) \leq L_{\max}(\bm{C}(\pi_1)) |\mathcal{I}]=1. \nonumber
\end{align}
By (i), $L_{\max}(\bm{V}(\text{EDD}_1))$ has the same distribution as $L_{\max}(\bm{V}(\text{EDD}))$. By (ii), $L_{\max}(\bm{C}(\pi_1))$ has the same distribution as $L_{\max}(\bm{C}(\pi))$.
Using the property of stochastic ordering \cite[Theorem 1.A.1]{StochasticOrderBook}, we can obtain \eqref{eq_delaygap1_thm3}. This completes the proof. 
\end{proof}

\begin{proof}[Proof of Theorem \ref{coro_thm3_1}]
By replacing Proposition \ref{lem2} with Proposition \ref{coro2} in the proof of Theorem \ref{thm3},
Theorem \ref{coro_thm3_1} is proven.
\end{proof}

\section{Conclusions}\label{sec_conclusion}
In this paper, we present a comprehensive study on near delay-optimal scheduling of batch jobs in queueing systems with parallel servers. We prove that for NBU task service time distributions and for arbitrarily given job parameters, the FUT policy, EDD  policy, and the  FCFS  policy are near delay-optimal in stochastic ordering for minimizing average delay, maximum lateness, and maximum delay, respectively, among all causal and non-preemptive policies. If the job parameters satisfy some additional conditions, we can further show that these policies are near delay-optimal in stochastic ordering for minimizing three general classes of delay metrics.

The key tools for showing these results are some novel sample-path conditions for comparing the delay performance of different policies. 
These sample-path conditions do not need to specify the queueing system model and hence can potentially be  applied to obtain near delay-optimal results for other queueing systems. One application of these sample-path conditions is in \cite{Yin_report2016}, which considered the problem of scheduling batch jobs in multi-server queueing systems with replications. 

\section*{Acknowledgements}

The authors are grateful to Ying Lei, Eytan Modiano, R. Srikant, and Junshan Zhang for helpful
discussions.

%
%
%

\ifreport
\appendices
\fi
\section{Proof of Lemma  \ref{lem_ratio}}\label{app_lem_ratio}
We will need the following lemma:
\begin{lemma}\label{lem_independent}
Suppose that $X_1, \ldots, X_m$ are non-negative independent random variables, $\chi_1, \ldots, \chi_m$ are arbitrarily given non-negative constants, $R_l = [X_l - \chi_l | X_l > \chi_l]$ for  $l=1,\ldots,m$, then $R_1, \ldots, R_m$  are  mutually independent.
\end{lemma}
\begin{proof}
For all constants $t_l \geq 0$, $l=1,\dots m$, we have
\begin{align}
&\Pr[R_l >t_l, l=1,\ldots,m] \nonumber\\
=& \Pr[X_l - \chi_l >t_l, l=1,\ldots,m| X_l > \chi_l, l=1,\ldots,m]\nonumber\\
=&\frac{\Pr[X_l   >t_l + \chi_l, l=1,\ldots,m]}{\Pr[X_l > \chi_l, l=1,\ldots,m]}\nonumber\\
=& \frac{\prod_{l=1}^m \Pr[X_l   >t_l + \chi_l]}{\prod_{l=1}^m \Pr[X_l   > \chi_l]}\nonumber\\
=& \prod_{l=1}^m \Pr[X_l - \chi_l >t_l| X_l > \chi_l]\nonumber\\
=&\prod_{l=1}^m \Pr[R_l >t_l].
\end{align}
Hence, $R_1, \ldots, R_m$  are mutually independent.
\end{proof}

\begin{proof}[Proof of lemma \ref{lem_ratio}]
Define a function $g_2(\bm{x})=f(\bm{x}+\bm{a})$.
Because $g_2(\bm{x})=f(\bm{x} + \bm{a})$ is sub-additive in $\bm{x}$, i.e., $g_2(\bm{x} + \bm{y}) \leq g_2(\bm{x} ) + g_2(\bm{y} )$, we can obtain
\begin{align}
& f(\bm{C}(P)) \nonumber\\
=& g_2(\left[\bm{V}(P)-\bm{a}\right]+\left[\bm{C}(P)-\bm{V}(P)\right]) \nonumber\\
\leq &  g_2(\bm{V}(P)-\bm{a})+g_2(\bm{C}(P)-\bm{V}(P))\nonumber\\
=& f(\bm{V}(P))+g_2(\bm{C}(P)-\bm{V}(P)).\label{eq_lem_ratio_3}
\end{align}
At time $V_i(P)$, all tasks of job $i$ have started service. If $k_i> m$, then job $i$ has at most $m$ incomplete tasks at time $V_i(P)$ because there are only $m$ servers. If $k_i\leq m$, then job $i$ has at most $k_i$ incomplete tasks at time $V_i(P)$. Therefore, at most $k_i\wedge m$ tasks of job $i$ are in service at time $V_i(P)$. Let $X_{i,l}$ denote the random service time of the $l$-th remaining task of job $i$. 
We will prove that for all $\mathcal{I}$
\begin{align}\label{eq_order_good}
[\bm{C}(P)-\bm{V}(P) | \mathcal{I}]\leq_{\text{st}}\left[\bigg(\max_{l=1,\ldots,k_1 \wedge m} X_{1,l},\ldots,\max_{l=1,\ldots,k_n \wedge m} X_{n,l}\bigg)\bigg| k_1,\ldots,k_n\right].
\end{align}
Without loss of generality, suppose that at time $V_i(P)$, the remaining tasks of job $i$ are being executed by the set of servers $\mathcal{S}_i\subseteq\{1,\ldots,m\}$, which satisfies $|\mathcal{S}_i|\leq k_i \wedge m$. Hence, the $X_{i,l}$'s are \emph{i.i.d.} In policy $P$, let $\chi_{i,l}$ denote the amount of time that was spent on executing the $l$-th remaining task of job $i$ before time $V_i(P)$, and let $R_{i,l}$ denote the remaining service time to complete this task after time $V_i(P)$. Then, $R_{i,l}$ can be expressed as $R_{i,l} = [X_{i,l} - \chi_{i,l} | X_{i,l} > \chi_{i,l}]$. Suppose that in policy $P$, $V_{j_i}(P)$ associated with job $j_i$ is the $i$-th smallest component of the vector $\bm{V}(P)$, i.e., $V_{j_i}(P)=V_{(i)}(P)$. We prove \eqref{eq_order_good} by using an inductive argument supported by Theorem 6.B.3 of \cite{StochasticOrderBook}, which contains two steps.

\emph{Step 1: We first show that for all $\mathcal{I}$} 
\begin{align}\label{eq_order_good1}
[{C}_{j_1}(P)-V_{j_1}(P)|\mathcal{I}] \leq_{\text{st}}  \left[\max_{{l=1,\ldots,k_{j_1} \wedge m}} X_{j_1,l}\bigg|\mathcal{I}\right].
\end{align}
Because the $X_{i,l}$'s are \emph{i.i.d.} NBU, for all realizations of $j_1$ and $\chi_{j_1,l}$, we have
\begin{align}
[R_{j_1,l}|j_1,\chi_{j_1,l}] \leq_{\text{st}} [X_{j_1,l}| j_1],~\forall~l\in \mathcal{S}_{j_1}.\label{eq_order_good4}
\end{align}
Lemma \ref{lem_independent} tells us that 
the $R_{j_1,l}$'s are conditional independent for any given realization of $\{\chi_{{j_1},l},l\in \mathcal{S}_{j_1}\}$. Hence, for all realizations of $j_1$, $\mathcal{S}_{j_1}$, and $\{\chi_{{j_1},l},l\in \mathcal{S}_{j_1}\}$
\begin{align}
&~~[{C}_{j_1}(P)-V_{j_1}(P)|{j_1}, \mathcal{S}_{j_1},\{\chi_{j_1,l}, l\in \mathcal{S}_{j_1}\}] \nonumber\\
&=  [\max_{l\in \mathcal{S}_{j_1}} R_{j_1,l} |{j_1},\mathcal{S}_{j_1},\{\chi_{j_1,l}, l\in \mathcal{S}_{j_1}\}] \nonumber\\
&\leq_{\text{st}}  [\max_{l\in \mathcal{S}_{j_1}} X_{j_1,l} |{j_1},\mathcal{S}_{j_1}]\label{eq_order_good2}\\
& \leq [\max_{l=1,\ldots,k_{j_1} \wedge m} X_{j_1,l} |{j_1}],\label{eq_order_good3}
\end{align}
where in \eqref{eq_order_good2} we have used \eqref{eq_order_good4} and Theorem 6.B.14 of \cite{StochasticOrderBook}, and in \eqref{eq_order_good3} we have used $|\mathcal{S}_i|\leq k_i \wedge m$. Because ${j_1}$, $\mathcal{S}_{j_1}$, and $\{\chi_{j_1,l}, l\in \mathcal{S}_{j_1}\}$
 are random variables which depend on  the job parameters $\mathcal{I}$, by using Theorem 6.B.16(e) of \cite{StochasticOrderBook},  \eqref{eq_order_good1} is proven.

\emph{Step 2: We show that for all given $\mathcal{I}$ and $\{{C}_{j_i}(P),V_{j_i}(P), i=1,\ldots h\}$} 
\begin{align}\label{eq_order_good6}
[{C}_{j_{h+1}}(P)-V_{j_{h+1}}(P)|\mathcal{I},\{{C}_{j_i}(P),V_{j_i}(P), i=1,\ldots h\}] \leq_{\text{st}}  \left[\max_{{l=1,\ldots,k_{j_{h+1}} \wedge m}} X_{j_{h+1},l}\bigg| \mathcal{I}\right].
\end{align}
Similar with \eqref{eq_order_good3}, for all realizations of $j_{h+1}$, $\mathcal{S}_{j_{h+1}}$, and $\{\chi_{j_{h+1},l}, l\in \mathcal{S}_{j_{h+1}}\}$, we can get
\begin{align}
[{C}_{j_{h+1}}(P)-V_{j_{h+1}}(P)|j_{h+1},\mathcal{S}_{j_{h+1}},\{\chi_{j_{h+1},l}, l\in \mathcal{S}_{j_{h+1}}\}] \leq_{\text{st}} [\max_{{l=1,\ldots,k_{j_{h+1}} \wedge m}} X_{j_{h+1},l}|{j_{h+1}}].\nonumber
\end{align}
Because $j_{h+1}$, $\mathcal{S}_{j_{h+1}}$, and $\{\chi_{j_{h+1},l}, l\in \mathcal{S}_{j_{h+1}}\}$
 are random variables which are determined by $\mathcal{I}$ and $\{{C}_{j_i}(P),V_{j_i}(P), i=1,\ldots h\}$, 
by using Theorem 6.B.16(e) of \cite{StochasticOrderBook}, \eqref{eq_order_good6} is proven.
By \eqref{eq_order_good1}, \eqref{eq_order_good6}, and Theorem 6.B.3 of \cite{StochasticOrderBook}, we have
 \begin{align}
&[\bm{C}(P)-\bm{V}(P) | \mathcal{I}]\nonumber\\
\leq_{\text{st}} &\left[\bigg(\max_{l=1,\ldots,k_1 \wedge m} X_{1,l},\ldots,\max_{l=1,\ldots,k_n \wedge m} X_{n,l}\bigg)\bigg| \mathcal{I}\right] \nonumber\\
=~ & \left[\bigg(\max_{l=1,\ldots,k_1 \wedge m} X_{1,l},\ldots,\max_{l=1,\ldots,k_n \wedge m} X_{n,l}\bigg)\bigg| k_1,\ldots,k_n\right].\nonumber
\end{align}
Hence, \eqref{eq_order_good} holds in policy $P$.

In policy $\pi$, $k_i$ tasks of job $i$ must be accomplished to complete job $i$. Hence, using the above arguments again, yields   
\begin{align}\label{eq_order_good7}
\left[\bigg(\max_{l=1,\ldots,k_1 \wedge m} X_{1,l},\ldots,\max_{l=1,\ldots,k_n \wedge m} X_{n,l}\bigg)\bigg| k_1,\ldots,k_n\right]\leq_{\text{st}} [\bm{C}(\pi) - \bm{a}| \mathcal{I}].
\end{align}
Combining \eqref{eq_order_good1} and \eqref{eq_order_good7}, we have
\begin{align}
[\bm{C}(P)-\bm{V}(P) | \mathcal{I}]\leq_{\text{st}} [\bm{C}(\pi) - \bm{a}| \mathcal{I}].\nonumber\end{align}
Because $g_2$ is increasing, by Theorem 6.B.16(a) of \cite{StochasticOrderBook}
\begin{align}\label{eq_order_good8}
[g_2(\bm{C}(P)-\bm{V}(P))|\mathcal{I}] \leq_{\text{st}}  [g_2(\bm{C}(\pi) - \bm{a})|\mathcal{I}] = [f(\bm{C}(\pi))|\mathcal{I}].
\end{align}
By substituting \eqref{eq_ratio_0} and \eqref{eq_order_good8} into \eqref{eq_lem_ratio_3}, \eqref{eq_ratio_1} is proven. This completes the proof.
\end{proof}

\section{Proof of Proposition~\ref{ordering_2}}\label{app_lem2}
Let $j$ be any integer chosen from $\{1,\ldots,n\}$, and $y_{j}$ be the number of jobs that have arrived by the time  $c_{(j)}(\pi)$, where $y_{j}\geq j$. Because $j$ jobs are completed by the time $c_{(j)}(\pi)$ in policy $\pi$, 
there are exactly $(y_{j}-j)$ incomplete jobs in the system at time $c_{(j)}(\pi)$. By the definition of the system state, we have $\xi_{[i],\pi}(c_{(j)}(\pi))=0$ for $i=y_{j}-j+1,\ldots,n$. Hence,
\begin{align}
\sum_{i=y_{j}-j+1}^n\xi_{[i],\pi}(c_{(j)}(\pi))=0.\nonumber
\end{align}
Combining this with \eqref{eq_ordering_2_1}, yields that policy $P$ satisfies 
\begin{align}\label{eq_proof_1}
\sum_{i=y_{j}-j+1}^n \gamma_{[i],P}\left(c_{(j)}(\pi)\right)\leq 0.
\end{align}
Further, the definition of the system state tells us that $ \gamma_{i,P}\left(t\right)\geq0$ holds for all $i=1,\ldots,n$ and $t\geq0$. Hence, we have
\begin{align}\label{eq_proof_2}
\gamma_{[i],P}\left(c_{(j)}(\pi)\right)=0, ~\forall~i=y_j-j+1,\ldots,n.
\end{align}
Therefore, there are at most $y_j-j$ jobs which have unassigned tasks at time  $c_{(j)}(\pi)$ in policy $P$.
Because  the sequence of job arrival times $a_1,a_2,\ldots,a_n$ are invariant under any policy, $y_j$ jobs have arrived by the time  $c_{(j)}(\pi)$ in policy $P$. Thus, there are at least $j$ jobs which have no unassigned tasks  at the time  $c_{(j)}(\pi)$  in policy $P$, which can be equivalently expressed as
\begin{align}\label{eq_proof_222}
v_{(j)}(P)\!  \leq c_{(j)}(\pi).
\end{align}
Because $j$ is arbitrarily chosen, \eqref{eq_proof_222} holds for all $j=1,\ldots,n$, which is exactly \eqref{eq_ordering_2_2}. In addition, \eqref{eq_ordering_2_3} follows from \eqref{eq_ordering_2_2}, which completes the proof.

\section{Proofs of Proposition~\ref{ordering_3} and Proposition~\ref{ordering_3_1}}\label{app_lem3}
\begin{proof}[Proof of Proposition~\ref{ordering_3}]
Let $w_i$ be the index of the job associated with the job completion time $c_{(i)}(P)$. In order to prove \eqref{eq_ordering_3_2}, it is sufficient to show that for each $j=1,2,\ldots,n$,
\begin{align}\label{eq_Condition_3_12_thm3}
c_{w_j}(P)-d_{w_j} \leq \max_{i=1,2,\ldots,n }[c_{i}(\pi)-d_i].
\end{align}
We prove \eqref{eq_Condition_3_12_thm3} by contradiction. \emph{For this, let us assume that 
\begin{align}\label{eq_proof_9_thm3}
c_{i}(\pi)< c_{w_j}(P) 
\end{align}
holds for all job $i$ satisfying $a_i\leq c_{w_j}(P)$ and $d_i\leq d_{w_j}$.} That is, if job $i$ arrives before time $c_{w_j}(P)$ and its due time is no later than $d_{w_j}$, then  job $i$ is completed before time $c_{w_j}(P)$ in policy $\pi$.
Define  
\begin{align}\label{eq_def}
\tau_j=\max_{i:a_i\leq c_{w_j}(P),d_i\leq d_{w_j}} c_{i}(\pi).
\end{align} 
According to \eqref{eq_proof_9_thm3} and \eqref{eq_def}, we can obtain
\begin{align}\label{eq_proof_19_thm3}
\tau_j<c_{w_j}(P).
\end{align}

On the other hand, \eqref{eq_def} tells us that all job  $i$ satisfying $d_i\leq d_{w_j}$ and $a_i\leq c_{w_j}(P)$ are completed by time $\tau_j$ in policy $\pi$. 
By this,
the system state of policy $\pi$ satisfies  
\begin{align}
\sum_{i:d_i\leq d_{w_j}}\xi_{i,\pi}(\tau_j) =0.\nonumber
\end{align}
Combining this with \eqref{eq_ordering_3_1}, yields
\begin{align}\label{eq_proof_1_thm3}
\sum_{i:d_i\leq d_{w_j}} \xi_{i,P}(\tau_j)\leq 0.
\end{align}
Further, the definition of the system state tells us that $\xi_{i,P}(t)\geq0$ for all $i=1,\ldots,n$ and $t\geq0$. Using this and \eqref{eq_proof_1_thm3}, we get that job $w_j$ satisfies
\begin{align}
\xi_{w_j,P}(\tau_j)= 0.\nonumber
\end{align}
That is, all tasks of job $w_j$ are completed by time $\tau_j$ in policy $P$. Hence, $c_{w_j}(P)\leq \tau_j$, where contradicts with \eqref{eq_proof_19_thm3}. Therefore, there exists at least one job $i$ satisfying the conditions 
$a_i\leq c_{w_j}(P)$, $d_i\leq d_{w_j}$, and $c_{w_j}(P) \leq c_{i}(\pi)$. This can be equivalently expressed as
\begin{align}\label{eq_proof_3_thm3}
c_{w_j}(P) \leq \max_{i:a_i\leq c_{w_j}(P),d_i\leq d_{w_j} }c_{i}(\pi).
\end{align}
Hence, for each $j=1,2,\ldots,n$,
\begin{align}
c_{w_j}(P)-d_{w_j} &\leq \max_{i:a_i\leq c_{w_j}(P),d_i\leq d_{w_j} }c_{i}(\pi)-d_{w_j} \nonumber\\
& \leq \max_{i:a_i\leq c_{w_j}(P),d_i\leq d_{w_j} }[c_{i}(\pi)-d_{i}]\nonumber\\
&\leq \max_{i=1,2,\ldots,n }[c_{i}(\pi)-d_i].\nonumber
\end{align}
This implies \eqref{eq_ordering_3_2}. Hence, Proposition~\ref{ordering_3} is proven.
\end{proof}
The proof of Proposition~\ref{ordering_3_1} is almost identical with that of 
Proposition~\ref{ordering_3}, and hence is not repeated here. The only difference is that $c_{w_j}(P)$ and $\bm{\xi}_{P}(\tau_j)$ in the proof of Proposition~\ref{ordering_3} should be replaced by $v_{w_j}(P)$ and $\bm{\gamma}_{P}(\tau_j)$, respectively.

\section{Proof of Proposition \ref{lem1}} \label{app0}

The following two lemmas are needed to prove Proposition \ref{lem1}:  

\begin{lemma}\label{lem_non_prmp1}
Suppose that under policy $P$, $\{\bm{\xi}_{P}',\bm{\gamma}_{P}'\}$ is obtained by allocating $b_P$ unassigned tasks to the servers in the system whose state is $\{\bm{\xi}_{P},\bm{\gamma}_{P}\}$. Further, suppose that under policy $\pi$, $\{\bm{\xi}_{\pi}',\bm{\gamma}_{\pi}'\}$ is obtained by completing $b_\pi$ tasks in the system whose state is $\{\bm{\xi}_{\pi},\bm{\gamma}_{\pi}\}$.
If $b_P\geq b_\pi$, condition 2 of Proposition \ref{lem1} is satisfied in policy $P$, and
\begin{eqnarray}\label{eq_non_prmp_41}
\sum_{i=j}^n {\gamma}_{[i],P}\leq \sum_{i=j}^n {\xi}_{[i],\pi}, ~\forall~j = 1,2,\ldots,n,\nonumber
\end{eqnarray}
then
\begin{eqnarray}\label{eq_non_prmp_40}
\sum_{i=j}^n {\gamma}_{[i],P}'\leq \sum_{i=j}^n {\xi}_{[i],\pi}', ~\forall~j = 1,2,\ldots,n.
\end{eqnarray}
\end{lemma}

\begin{proof}
If $\sum_{i=j}^n {\gamma}_{[i],P}'=0$, then the inequality \eqref{eq_non_prmp_40} follows naturally. 
If $\sum_{i=j}^n {\gamma}_{[i],P}'>0$, then there exist  unassigned tasks which have not been assigned to any server. 
In policy $P$, each task allocated to the servers is from the job with the minimum positive ${\gamma}_{i,P}$. Hence,
$\sum_{i=j}^n {\gamma}_{[i],P}'=\sum_{i=j}^n {\gamma}_{[i],P} - b_P \leq \sum_{i=j}^n {\xi}_{[i],\pi} - b_\pi \leq \sum_{i=j}^n {\xi}_{[i],\pi}'$.
\end{proof}

\begin{lemma}\label{lem_non_prmp2}
Suppose that, under policy $P$, $\{\bm{\xi}_{P}',\bm{\gamma}_{P}'\}$ is obtained by adding a job with $b$ tasks to the system whose state is $\{\bm{\xi}_{P},\bm{\gamma}_{P}\}$. Further, suppose that, under policy $\pi$, $\{\bm{\xi}_{\pi}',\bm{\gamma}_{\pi}'\}$ is obtained by adding a job with $b$ tasks to the system whose state is $\{\bm{\xi}_{\pi},\bm{\gamma}_{\pi}\}$.
If
\begin{eqnarray}
\sum_{i=j}^n {\gamma}_{[i],P}\leq \sum_{i=j}^n {\xi}_{[i],\pi}, ~\forall~j = 1,2,\ldots,n,\nonumber
\end{eqnarray}
then
\begin{eqnarray}
\sum_{i=j}^n {\gamma}_{[i],P}'\leq \sum_{i=j}^n {\xi}_{[i],\pi}', ~\forall~j = 1,2,\ldots,n.\nonumber
\end{eqnarray}
\end{lemma}

\begin{proof}
Without loss of generalization, we suppose that after the job arrival, $b$ is the $l$-th largest component of $\bm{\gamma}_{P}'$ and the $m$-th largest component of $\bm{\xi}_{\pi}'$, i.e., $\gamma'_{[l],P} = \xi'_{[m],\pi}= b$. We consider the following four cases:

{Case 1}: $l<j, m<j$. We have $\sum_{i=j}^n {\gamma}_{[i],P}' =\sum_{i=j-1}^n {\gamma}_{[i],P} \leq \sum_{i=j-1}^n {\xi}_{[i],\pi}= \sum_{i=j}^n {\xi}_{[i],\pi}'$.

{Case 2}: $l<j, m\geq j$. We have $\sum_{i=j}^n {\gamma}_{[i],P}' =\sum_{i=j-1}^n {\gamma}_{[i],P} \leq b + \sum_{i=j}^n {\gamma}_{[i],P} \leq b + \sum_{i=j}^n {\xi}_{[i],\pi} = \sum_{i=j}^n {\xi}_{[i],\pi}'$.

{Case 3}: $l\geq j, m<j$. We have $\sum_{i=j}^n {\gamma}_{[i],P}' = b + \sum_{i=j}^n {\gamma}_{[i],P} \leq \sum_{i=j-1}^n {\gamma}_{[i],P} \leq \sum_{i=j-1}^n {\xi}_{[i],\pi} = \sum_{i=j}^n {\xi}_{[i],\pi}'$.

{Case 4}: $l\geq j, m\geq j$. We have $\sum_{i=j}^n {\gamma}_{[i],P}' = b + \sum_{i=j}^n {\gamma}_{[i],P} \leq b + \sum_{i=j}^n {\xi}_{[i],\pi} = \sum_{i=j}^n {\xi}_{[i],\pi}'$.
\end{proof}

We now use Lemma \ref{lem_non_prmp1} and Lemma \ref{lem_non_prmp2} to prove Proposition \ref{lem1}.
\ifreport
\begin{proof}[Proof of Proposition \ref{lem1}]
\else
\begin{proof}[of Proposition \ref{lem1}]
\fi

Assume that no task is completed at the job arrival times $a_i$ for $i=1,\ldots,n$. This does not lose any generality, because if a task is completed at time $t_j=a_i$, Proposition \ref{lem1} can be proven by first proving for the case $t_j=a_i+\epsilon$ and then taking the limit $\epsilon\rightarrow 0$. 
We prove \eqref{eq_ordering_2_1} by induction. 

\emph{Step 1: We will show that \eqref{eq_ordering_2_1} holds during $[0,a_2)$.\footnote{Note that $a_1=0$.}} 

Because $\bm{\xi}_{P}(0^-) = \bm{\gamma}_{P}(0^-) =\bm{\xi}_{\pi}(0^-) = \bm{\gamma}_{\pi}(0^-)=\bm{0}$, \eqref{eq_ordering_2_1} holds at time $0^-$. Job 1 arrives at time $a_1=0$. By Lemma \ref{lem_non_prmp2}, \eqref{eq_ordering_2_1} holds at time $0$. Let $t$ be an arbitrarily chosen time during $(0,a_{2})$. Suppose that  $b_\pi$ tasks start execution and also complete execution during $[0,t]$ in policy $\pi$. We need to consider two cases: 

Case 1: The queue is not empty (there exist unassigned tasks in the queue) during $[0,t]$ in policy $P$. By the weak work-efficiency ordering condition, no fewer than $b_\pi$ tasks start execution during $[0,t]$ in policy $P$. Because \eqref{eq_ordering_2_1} holds at time $0$, by Lemma \ref{lem_non_prmp1}, \eqref{eq_ordering_2_1} also holds at time $t$. 

Case 2: The queue is empty (all tasks in the system are in service) by time $t'\in[0,t]$ in policy $P$. Because $t\in(0,a_{2})$ and there is no task arrival during $(0,a_2)$, there is no task arrival during $(t',t]$. Hence, it must hold that all tasks in the system are in service at time $t$. Then, the system state of policy $P$ satisfies $\sum_{i=j}^n {\gamma}_{[i],P}(t)=0$ for all $j=1,2,\ldots,n$ at time $t$. Hence, \eqref{eq_ordering_2_1} holds at time $t$.

In summary of these two cases, \eqref{eq_ordering_2_1} holds for all $t\in[0,a_2)$.

\emph{Step 2: Assume that for some integer $i\in\{2,\ldots, n\}$, the conditions of Proposition \ref{lem1} imply that \eqref{eq_ordering_2_1} holds for all $t\in[0,a_i)$. We will prove that the conditions of Proposition \ref{lem1} imply that  \eqref{eq_ordering_2_1} holds for all $t\in[0,a_{i+1})$.}

Let $t$ be an arbitrarily chosen time during $(a_i,a_{i+1})$. We modify the task completion times in policy $\pi$ as follows: For each pair of corresponding task $j$ and task $j'$ mentioned in the definition of the weak work-efficiency ordering, if 
\begin{itemize}
\item In policy $\pi$, task $j$ starts execution at time $\tau\in[0,a_i]$ and completes execution at time $\nu\in(a_i,t]$,
\item In policy $P$, the queue is not empty (there exist unassigned tasks in the queue) during $[\tau,\nu]$,
\item In policy $P$, the corresponding task $j'$ starts execution at time $t'\in[0,a_i]$,
\end{itemize}
then the completion time of task $j$ is modified from $\nu$ to $a_i^-$ in policy $\pi$, as illustrated in Fig. \ref{fig_modification}.

This modification satisfies the following three claims:
\begin{itemize}
\item[1.] The system state of policy $\pi$ at time $t$ remains the same before and after this modification;
\item[2.] Policy $P$ is still weakly more work-efficient than policy $\pi$ after this modification;
\item[3.] If $b_\pi$ tasks complete execution during $[a_i,t]$ on the modified sample path of policy $\pi$, and  the queue is not empty (there exist unassigned tasks in the queue) during $[a_i,t]$ in policy $P$, then no fewer than $b_\pi$ tasks start execution during $[a_i,t]$ in policy $P$.
\end{itemize}
We now prove these three claims. Claim 1 follows from the fact that the tasks completed during $[0,t]$ remain the same before and after this modification. It is easy to prove Claim 2 by checking the definition of work-efficiency ordering. For Claim 3, notice that if a task $j$ starts execution and completes execution during $[a_i,t]$ on the modified sample path of policy $\pi$, then by Claim 2, its corresponding task $j'$ must start execution during $[a_i,t]$ in policy $P$. On the other hand, if a task $j$ starts execution during $[0,a_i]$ and completes execution during $[a_i,t]$ on the modified sample path of policy $\pi$, then by the modification, its corresponding task $j'$ must start execution during $[a_i,t]$ in policy $P$. By combining these two cases, Claim 3 follows.

\ifreport
\begin{figure}
\centering 
\includegraphics[width=0.7\textwidth]{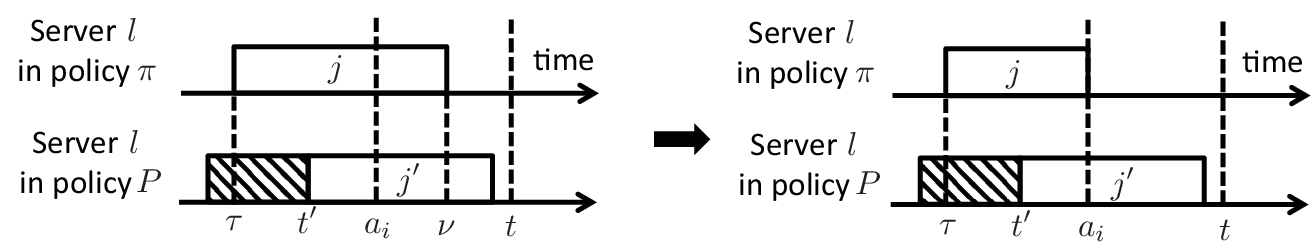} \caption{Illustration of the modification of task completion times in policy $\pi$: If in policy $\pi$, task $j$ starts execution at time $\tau\in[0,a_i]$ and completes execution at time $\nu\in(a_i,t]$, and in policy $P$, task $j'$ starts execution at time $t'\in[0,a_i]$, then the completion time of task $j$ is changed from $\nu$ to $a_i^-$ in policy $\pi$.}
\label{fig_modification} 
\end{figure} 
\fi

We use these three claims to prove the  statement of \emph{Step 2}. According to Claim 2, policy $P$ is weakly more work-efficient than policy $\pi$ after the modification. By the assumption of \emph{Step 2}, \eqref{eq_ordering_2_1} holds during $[0,a_i)$ for the modified sample path of policy $\pi$. Job $j$ arrives at time $a_i$. By Lemma \ref{lem_non_prmp2}, \eqref{eq_ordering_2_1} holds at time $a_i$ for the modified sample path of policy $\pi$. Suppose that $b_\pi$ tasks complete execution during $[a_i,t]$ on the modified sample path of policy $\pi$.
We need to consider two cases: 

Case 1: The queue is not empty (there exist unassigned tasks in the queue) during $[a_i,t]$ in policy $P$. By Claim 3, no fewer than $b_\pi$ tasks start execution during $[a_i,t]$ in policy $P$. Because \eqref{eq_ordering_2_1} holds at time $a_i$, by Lemma \ref{lem_non_prmp1}, \eqref{eq_ordering_2_1} also holds at time $t$ for the modified sample path of policy $\pi$. 

Case 2: The queue is empty (all tasks in the system are in service) at  time $t'\in[a_i,t]$ in policy $P$. Because $t\in(a_i,a_{i+1})$ and $t'\in[a_i,t]$, there is no task arrival during $(t',t]$. Hence, it must hold that all tasks in the system are in service at time $t$. Then, the system state of policy $P$ satisfies $\sum_{i=j}^n {\gamma}_{[i],P}(t)=0$ for all $j=1,2,\ldots,n$  at time $t$. Hence, \eqref{eq_ordering_2_1} holds at time $t$ for the modified sample path of policy $\pi$.

In summary of these two cases, \eqref{eq_ordering_2_1} holds at time $t$ for the modified sample path of policy $\pi$. By Claim 1, the system state of policy $\pi$ at time $t$ remains the same before and after this modification. Hence, \eqref{eq_ordering_2_1} holds at time $t$ for the original sample path of policy $\pi$. Therefore, if  the assumption of \emph{Step 2} is true, then \eqref{eq_ordering_2_1} holds for all $t\in[0,a_{i+1})$. 

By induction, \eqref{eq_ordering_2_1} holds at time $t\in[0,\infty)$. Then, \eqref{eq_ordering_2_2} and \eqref{eq_ordering_2_3} follow from Proposition \ref{ordering_2}.
This completes the proof.\end{proof}


\section{Proof  of Proposition \ref{lem2}} \label{app0_1}



The proof of Proposition \ref{lem2} requires the following two lemmas:

\begin{lemma}\label{lem_non_prmp1_thm3}
Suppose that, in policy $P$, $\{\bm{\xi}_P',\bm{\gamma}_P'\}$ is obtained by allocating $b_P$ unassigned tasks to the servers in the system whose state is $\{\bm{\xi}_P,\bm{\gamma}_P\}$. Further, suppose that, in policy $\pi$, $\{\bm{\xi}_\pi',\bm{\gamma}_\pi'\}$ is obtained by completing $b_\pi$ tasks in the system whose state is $\{\bm{\xi}_\pi,\bm{\gamma}_\pi\}$.
If $b_P\geq b_\pi$, condition 2 of Proposition \ref{lem2} is satisfied in policy $P$, and
\begin{eqnarray}
\sum_{i:d_i\leq\tau} \gamma_{i,P}\leq \sum_{i:d_i\leq\tau} \xi_{i,\pi},~\forall~\tau\in[0,\infty),\nonumber
\end{eqnarray}
then
\begin{eqnarray}\label{eq_non_prmp_40_thm3}
\sum_{i:d_i\leq\tau} \gamma_{i,P}'\leq \sum_{i:d_i\leq\tau} \xi_{i,\pi}',~\forall~\tau\in[0,\infty).\end{eqnarray}
\end{lemma}

\begin{proof}
If $\sum_{i:d_i\leq\tau} \gamma_{i,P}'=0$, then the inequality \eqref{eq_non_prmp_40_thm3} follows naturally. 
If $\sum_{i:d_i\leq\tau} \gamma_{i,P}'>0$, then there exist some unassigned tasks in the queue. 
In policy $P$, each task allocated to the servers is from the job with the earliest due time. Hence,
$\sum_{i:d_i\leq\tau} \gamma_{i,P}'=\sum_{i:d_i\leq\tau} \gamma_{i,P} - b_P \leq \sum_{i:d_i\leq\tau} \xi_{i,\pi} -b_\pi \leq \sum_{i:d_i\leq\tau} \xi_{i,\pi}'$.
\end{proof}

\begin{lemma}\label{lem_non_prmp2_thm3}
Suppose that under policy $P$, $\{\bm{\xi}_P',\bm{\gamma}_P'\}$ is obtained by adding a job with $b$ tasks and due time $d$ to the system whose state is $\{\bm{\xi}_P,\bm{\gamma}_P\}$. Further, suppose that under policy $\pi$, $\{\bm{\xi}_\pi',\bm{\gamma}_\pi'\}$ is obtained by adding a job with $b$ tasks and due time $d$ to the system whose state is $\{\bm{\xi}_\pi,\bm{\gamma}_\pi\}$.
If
\begin{eqnarray}
\sum_{i:d_i\leq\tau} \gamma_{i,P}\leq \sum_{i:d_i\leq\tau} \xi_{i,\pi}, ~\forall~\tau\in[0,\infty),\nonumber
\end{eqnarray}
then
\begin{eqnarray}
\sum_{i:d_i\leq\tau} \gamma_{i,P}'\leq \sum_{i:d_i\leq\tau} \xi_{i,\pi}', ~\forall~\tau\in[0,\infty).\nonumber\end{eqnarray}
\end{lemma}

\begin{proof}
If $d\leq\tau$, then
$\sum_{i:d_i\leq\tau} \gamma_{i,P}'\leq \sum_{i:d_i\leq\tau} \gamma_{i,P} + b\leq \sum_{i:d_i\leq\tau} \xi_{i,\pi}+b \leq \sum_{i:d_i\leq\tau} \xi_{i,\pi}'$.

If $d>\tau$, then
$\sum_{i:d_i\leq\tau} \gamma_{i,P}'\leq \sum_{i:d_i\leq\tau} \gamma_{i,P} \leq \sum_{i:d_i\leq\tau} \xi_{i,\pi} \leq \sum_{i:d_i\leq\tau} \xi_{i,\pi}'$.
\end{proof}

The proof of Proposition \ref{lem2} is almost identical with that of Proposition \ref{lem1}, and hence is not repeated here. The only difference is that Lemma \ref{lem_non_prmp1} and Lemma \ref{lem_non_prmp2} in the proof of Proposition \ref{lem1} should be replaced by Lemma \ref{lem_non_prmp1_thm3} and Lemma \ref{lem_non_prmp2_thm3}, respectively.

\section{Proof of Proposition \ref{lem1_1}} \label{app_lem1_1}

We have proven that \eqref{eq_ordering_2_2} holds under the conditions of Proposition \ref{lem1_1}.
Note that \eqref{eq_ordering_2_2} can be equivalently expressed in the following vector form:
\begin{align}
\bm{v}_{\uparrow} (P)\leq \bm{c}_{\uparrow} (\pi).\nonumber
\end{align}
Because any $f\in\mathcal{D}_{\text{sym}}$ is a symmetric and increasing function, we can obtain
\begin{align}
&f(\bm{v} (P))=f(\bm{v}_{\uparrow} (P)) \nonumber\\
\leq_{}&f(\bm{c}_{\uparrow} (\pi))= f(\bm{c} (\pi)).\nonumber
\end{align} 
This completes the proof.


\section{Proof of Proposition \ref{lem_general}}\label{app_coro2}

In the proof of Proposition \ref{lem_general}, we need to use the following rearrangement inequality: 
\begin{lemma} \cite[Theorem 6.F.14]{Marshall2011}\label{lem_rearrangement}
Consider two $n$-dimensional vectors $(x_1,\dots,x_n)$ and $(y_1,\ldots,y_n)$. If $(x_i- x_j)(y_i-y_j)\leq0$ for two indices $i$ and $j$ where $1\leq i<j\leq n$, then 
\begin{align}
&(x_1\!-\!y_1,\ldots,x_j\!-\!y_i,\ldots,x_i\!-\!y_j,\ldots,x_n\!-\!y_n)\nonumber\\
\prec& (x_1\!-\!y_1,\ldots,x_i\!-\!y_i,\ldots,x_j\!-\!y_j,\ldots,x_n\!-\!y_n).\nonumber
\end{align}
\end{lemma}
\ifreport
\begin{proof}[Proof of Proposition \ref{lem_general}]
\else
\begin{proof}[of Proposition \ref{lem_general}]
\fi

For  $f\in\mathcal{D}_{\text{sym}}$, \eqref{eq_ordering_2_2} and  \eqref{eq_lem_general} follow from Proposition \ref{lem1} and Proposition \ref{lem1_1}.

For  $f\in\mathcal{D}_{\text{Sch-1}}$, \eqref{eq_lem_general} is proven in 3 steps, which are described as follows:



\emph{Step 1: We will show that}
\begin{align}\label{eq_coro2_3}
\bm{v}(P)-\bm{d}\prec_{\text{w}} \bm{c} (\pi)-\bm{d}.
\end{align} 
According to Eq. (1.A.17) and Theorem 5.A.9 of \cite{Marshall2011},
it is sufficient to show that there exists an $n$-dimensional vector $\bm{v}'$ such that 
\begin{align}\label{eq_coro2_6}
\bm{v}(P)-\bm{d}\prec \bm{v}'- \bm{d} \leq \bm{c} (\pi)-\bm{d}.
\end{align}
Vector $\bm{v}'$ is constructed as follows: First, the components of the vector $\bm{v}'$ is a rearrangement (or permutation) of the components of the vector $\bm{v}(P)$, which can be equivalently expressed as
\begin{align}\label{eq_coro2_4}
v_{(i)}'=v_{(i)}(P),~\forall~i=1,\ldots,n.
\end{align}
Second, for each $j=1,\ldots,n$, if the completion time $c_j(\pi)$ of job $j$  is the $i_j$-th smallest component of $\bm{c}(\pi)$, i.e., 
\begin{align}
c_j(\pi)= c_{(i_j)}(\pi),
\end{align}
then $v_j'$ associated with job $j$ is the $i_j$-th smallest component of $\bm{v}'$, i.e., 
\begin{align}\label{eq_coro2_5}
v_j' = v_{(i_j)}'.
\end{align}
Combining \eqref{eq_ordering_2_2} and \eqref{eq_coro2_4}-\eqref{eq_coro2_5}, yields
\begin{align}
v_j' = v_{(i_j)}'=v_{(i_j)}(P) \leq c_{(i_j)}(\pi) = c_j(\pi)\nonumber
\end{align}
for $j=1,\ldots,n$. This implies $\bm{v}' \leq \bm{c} (\pi)$, and hence the second inequality in \eqref{eq_coro2_6} is proven.

The remaining task is to prove the first inequality in \eqref{eq_coro2_6}. 
First, consider the case that the due times $d_1,\ldots, d_n$ of the $n$ jobs are  different from each other. 
The vector  $\bm{v}(P)$ can be obtained from  $\bm{v}'$ by the following procedure: For each $j=1,\ldots,n$, define a set 
\begin{align}
S_j = \{i: a_i\leq v_{j}(P), d_i < d_{j}\}.
\end{align}
If there exist two jobs $i$ and $j$ which satisfy $i\in S_j$ and $v_i' > v_{j}'$, we interchange the components $v_i'$ and $v_{j}'$ in vector $\bm{v}'$. Repeat this interchange operation, until such two jobs $i$ and $j$ satisfying  $i\in S_j$ and $v_i' > v_{j}'$ cannot be found. Therefore, at the end of this procedure, if job $i$ arrives before $v_{j}(P)$ and job $i$ has an earlier due time than job $j$, then $v_i' < v_{j}'$, which is exactly the priority rule of job service satisfied by policy $P$. Therefore, the vector  $\bm{v}(P)$ is obtained at the end of this procedure.
In each interchange operation of this procedure, $(v_i' - v_{j}')(d_i - d_{j})\leq0$ is satisfied before the interchange of $v_i'$ and $v_{j}'$. By Lemma \ref{lem_rearrangement} and the transitivity of the ordering of majorization, we can obtain $\bm{v}(P)-\bm{d}\prec \bm{v}'- \bm{d}$, which is the first inequality in \eqref{eq_coro2_6}. 

Next, consider the case that two jobs $i$ and $j$ have identical due time $d_i = d_j$. Hence, $(v_i' - v_{j}')(d_i - d_{j})=0$. In this case, the service order of job $i$ and job $j$ are indeterminate in policy $P$. Nonetheless, by Lemma \ref{lem_rearrangement}, the service order of job $i$ and job $j$ does not affect the first inequality in \eqref{eq_coro2_6}. Hence, the first inequality in \eqref{eq_coro2_6} holds even when $d_i = d_j$.
 
Finally, \eqref{eq_coro2_3} follows from \eqref{eq_coro2_6}.

\emph{Step 3: We use \eqref{eq_coro2_3} to prove Proposition \ref{lem_general}.}
For any $f\in\mathcal{D}_{\text{Sch-1}}$, $f(\bm{x}+\bm{d})$ is increasing and Schur convex. 
According to Theorem 3.A.8 of \cite{Marshall2011}, for all $f\in\mathcal{D}_{\text{Sch-1}}$, we have
\begin{align}
&f(\bm{v}(P)) \nonumber\\
= &f[(\bm{v}(P)-\bm{d})+\bm{d}] \nonumber\\
\leq& f[(\bm{c}(\pi)-\bm{d})+\bm{d}] \nonumber\\
=& f(\bm{c} (\pi)).\nonumber
\end{align} 
This completes the proof.
\end{proof}

\begin{figure}
\centering
\includegraphics[width=0.45\textwidth]{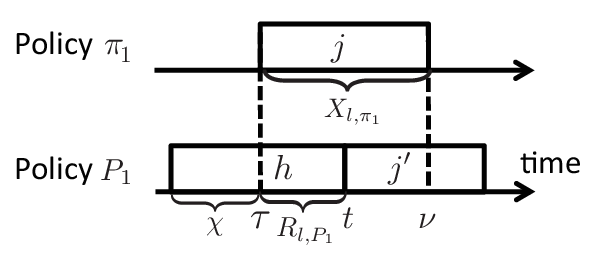} \caption{Illustration of the weak work-efficiency ordering between policy $\pi_1$ and policy $P_1$.}
\label{fig_theorem_proof_1}
\end{figure}
\section{Proof of Lemma \ref{lem_coupling}}\label{app1}
We use coupling to prove Lemma \ref{lem_coupling}: We construct two policies $P_1$ and $\pi_1$ such that policy $P_1$
satisfies the same queueing discipline with policy $P$, and policy $\pi_1$ satisfies the same queueing discipline with policy $\pi$. Hence, policy $P_1$ is work-conserving.
The task and job completion times of policy $P_1$ (policy $\pi_1$) have the same distribution as those of policy $P$ (policy $\pi$). Because the state process is determined by the job parameters $\mathcal{I}$ and the task/job completion events, the state process $\{\bm{\xi}_{P_1}(t),\bm{\gamma}_{P_1}(t),t\in[0,\infty)\}$ of policy $P_1$ has the same distribution as the state process $\{\bm{\xi}_{P}(t),\bm{\gamma}_{P}(t),t\in[0,\infty)\}$ of policy $P$, and the state process $\{\bm{\xi}_{\pi_1}(t),\bm{\gamma}_{\pi_1}(t),t\in[0,\infty)\}$ of policy $\pi_1$ has the same distribution as the state process $\{\bm{\xi}_{\pi}(t),\bm{\gamma}_{\pi}(t),t\in[0,\infty)\}$  of policy $\pi$.
 
Next, we show that policy $P_1$ and policy $\pi_1$ can be constructed such that policy $P_1$ is weakly more work-efficient than policy $\pi_1$ with probability one. Let us consider any task $j$ executed in policy $\pi_1$. Suppose that in policy $\pi_1$ task $j$ starts service in server $l$ at time $\tau$ and completes service at time $\nu$, and the queue is not empty during $[\tau,\nu]$ in policy $P_1$. Because policy $P_1$ is work-conserving, all servers are busy in policy $P_1$ during $[\tau,\nu]$. In particular, server $l$ is busy in policy $P_1$ during $[\tau,\nu]$. Suppose that  in policy $P_1$, server $l$ has spent a time duration $\chi$ ($\chi\geq0$) on executing a task $h$ before time $\tau$. Let $R_{l,P_1}$ denote the remaining service time of server $l$ for executing task $h$ after time $\tau$ in policy $P_1$. Let $X_{l,\pi_1}=\nu-\tau$ denote the service time of task $j$ in policy $\pi_1$ and $X_{l,P_1}=\chi+R_{l,P_1}$ denote the service time of task $h$ in policy $P_1$. 
Then, the complementary CDF of $R_{l,P_1}$  is given by
\begin{align}
\Pr[R_{l,P_1}>s] = \Pr[X_{l,P_1}-\chi>s|X_{l,P_1}> \chi ].\nonumber
\end{align}
Because the task service times are NBU, we can obtain that for all $s,\chi\geq 0$ 
\begin{align}
\Pr[X_{l,P_1}-\chi>s|X_{l,P_1}> \chi ] =\Pr[X_{l,\pi_1}-\chi>s|X_{l,\pi_1}> \chi ] \leq \Pr[X_{l,\pi_1}>s],\nonumber
\end{align}
and hence $R_{l,P_1} \leq_{\text{st}} X_{l,\pi_1}$.
By Theorem 1.A.1 of \cite{StochasticOrderBook} (i.e., a constructive definition of stochastic ordering), the random variables $R_{l,P_1}$ and $X_{l,\pi_1}$ can be constructed such that $R_{l,P_1} \leq X_{l,\pi_1}$ with probability one. That is, in policy $P_1$ server $l$ completes executing task $h$ at time $t=\tau + R_{l,P_1}$, which is earlier than time $\nu = \tau + X_{l,\pi_1}$ with probability one. Because in policy $P_1$ server $l$ is kept busy during $[\tau,\nu]$, a new task, say task $j'$, will start execution on server $l$ at time $t\in [\tau,\nu]$ with probability one. 

In the above coupling arguments, conditioned on every possible realization of policy $P_1$ and policy $\pi_1$ before the service of task $j$ starts, we can construct the service of task $j$ in policy $\pi_1$ and the service of the corresponding task $j'$ in policy $P_1$ such that the requirement of weak work-efficiency ordering is satisfied for this pair of tasks. Next, following the proof of \cite[Theorem 6.B.3]{StochasticOrderBook}, one can continue this procedure to progressively construct the service of all tasks in policy $\pi_1$ and policy $P_1$. By this, we obtain that policy $P$$_1$ is weakly more work-efficient than policy $\pi_1$ with probability one, which completes the proof.

\section{Proof of Theorem \ref{lem7_NBU}}\label{app_lem7}

Consider the time difference ${C}_i(\text{FUT})-V_i(\text{FUT})$.
At time $V_i(\text{FUT})$, all tasks of job $i$ are completed or in service.
if $k_i>m$,  then job $i$ has at most $m$ incomplete tasks that are in service at time $V_i(\text{FUT})$, see Fig. \ref{V_i} for an illustration; if $k_i\leq m$, then job $i$ has at most $k_i$ incomplete tasks that are in service  at time $V_i(\text{FUT})$. Therefore, in policy FUT, no more than $k_i\wedge m = \min\{k_i,m\}$ tasks of job $i$ are completed during the time interval $[V_i(\text{FUT}), {C}_i(\text{FUT})]$. 

Suppose that at time $V_i(\text{FUT})$, the remaining tasks of job $i$ are being executed by the set of servers  $\mathcal{S}_i\subseteq\{1,\ldots,m\}$, which satisfies $|\mathcal{S}_i|\leq k_i \wedge m$. Let $\chi_l$ denote the amount of time that server $l\in \mathcal{S}_i$ has spent on executing a task of job $i$ by time $V_i(\text{FUT})$ in policy FUT. Let $R_{l}$ denote the remaining service time of server $l\in \mathcal{S}_i$ for executing this task after time $V_i(\text{FUT})$. Then, $R_l$ can be expressed as $R_l = [X_l - \chi_l | X_l > \chi_l]$. Because the $X_l$'s are independent NBU random variables with mean $\mathbb{E}[X_l]=1/\mu_l$, for all realizations of $\chi_l$
\begin{align}
[R_l|\chi_l] \leq_{\text{st}} X_l,~\forall~l\in \mathcal{S}_i.\nonumber
\end{align}
In addition, Theorem 3.A.55 of \cite{StochasticOrderBook} tells us that
\begin{align}
X_l \leq_{\text{icx}} Z_l,~\forall~l\in \mathcal{S}_i,\nonumber
\end{align}
where $\leq_{\text{icx}}$ is the increasing convex order defined in \cite[Chapter 4]{StochasticOrderBook} and the $Z_l$'s are independent exponential random variables with mean $\mathbb{E}[Z_l]=\mathbb{E}[X_l] =\mu_l$. Hence, 
\begin{align}
[R_l|\chi_l] \leq_{\text{icx}} Z_l,~\forall~l\in \mathcal{S}_i.\nonumber
\end{align}
Lemma \ref{lem_independent} tells us that 
the $R_l$'s are conditional independent for any given realization of $\{\chi_l,l\in \mathcal{S}_i\}$. Hence, by Corollary 4.A.16 of \cite{StochasticOrderBook}, for all realizations of $\mathcal{S}_i$ and $\{\chi_l,l\in \mathcal{S}_i\}$
\begin{align}\label{eq_exp_order}
\big[\max_{l\in \mathcal{S}_i} R_l \big| \mathcal{S}_i, \{\chi_l, l\in \mathcal{S}_i\}\big] \leq_{\text{icx}} \big[\max_{l\in \mathcal{S}_i} Z_l\big| \mathcal{S}_i\big].
\end{align}
Then,
\begin{align}
&\mathbb{E}[{C}_i(\text{FUT})-V_i(\text{FUT})|\mathcal{S}_i ,\{\chi_l, l\in \mathcal{S}_i\}]\nonumber\\
=& \mathbb{E}\!\left[ \max_{l\in \mathcal{S}_i} R_l \bigg| \mathcal{S}_i ,\{\chi_l, l\in \mathcal{S}_i\}\right]\nonumber\\
\leq&\mathbb{E}\!\left[ \max_{l\in \mathcal{S}_i} Z_l \bigg| \mathcal{S}_i \right]\label{eq_gap_condition3}\\
\leq &\mathbb{E}\!\left[\max_{l=1,\ldots,k_i \wedge m} Z_l\right]\label{eq_gap_condition4}\\
\leq &\sum_{l=1}^{k_i \wedge m} \frac{1}{\sum_{j=1}^l \mu_j},\label{eq_gap_condition5}
\end{align}
where 
\eqref{eq_gap_condition3} is due to
\eqref{eq_exp_order} and Eq. (4.A.1) of \cite{StochasticOrderBook},
\eqref{eq_gap_condition4} is due to $\mu_1\leq \ldots\leq \mu_M$, $|\mathcal{S}_i|\leq k_i \wedge m$, and the fact that $\max_{l=1,\ldots,k_i \wedge m} Z_l$ is independent of $\mathcal{S}_i$, and \eqref{eq_gap_condition5} is due to the property of exponential distributions. Because $\mathcal{S}_{i}$ and $\{\chi_{l}, l\in \mathcal{S}_{i}\}$
 are random variables which are determined by the job parameters $\mathcal{I}$, taking the conditional expectation for given $\mathcal{I}$ in \eqref{eq_gap_condition5}, yields
\begin{align} 
\mathbb{E}\!\left[ {C}_i(\text{FUT})-V_i(\text{FUT}) | \mathcal{I}\right] \leq \sum_{l=1}^{k_i \wedge m} \frac{1}{\sum_{j=1}^l \mu_j}.\nonumber
\end{align} 
By taking the average over all $n$ jobs,  \eqref{eq_gap} is proven.
In addition, it is known that for each $k=1,2,\ldots,$
\begin{align}
\sum_{l=1}^k \frac{1}{l} \leq \ln(k)+1.\nonumber
\end{align}
By this, \eqref{eq_gap_2} holds. This completes the proof.
\bibliographystyle{IEEEtran}
\bibliography{ref}

\begin{thebibliography}{10}
\providecommand{\url}[1]{#1}
\csname url@samestyle\endcsname
\providecommand{\newblock}{\relax}
\providecommand{\bibinfo}[2]{#2}
\providecommand{\BIBentrySTDinterwordspacing}{\spaceskip=0pt\relax}
\providecommand{\BIBentryALTinterwordstretchfactor}{4}
\providecommand{\BIBentryALTinterwordspacing}{\spaceskip=\fontdimen2\font plus
\BIBentryALTinterwordstretchfactor\fontdimen3\font minus
  \fontdimen4\font\relax}
\providecommand{\BIBforeignlanguage}[2]{{%
\expandafter\ifx\csname l@#1\endcsname\relax
\typeout{** WARNING: IEEEtran.bst: No hyphenation pattern has been}%
\typeout{** loaded for the language `#1'. Using the pattern for}%
\typeout{** the default language instead.}%
\else
\language=\csname l@#1\endcsname
\fi
#2}}
\providecommand{\BIBdecl}{\relax}
\BIBdecl

\bibitem{mapreduce}
J.~Dean and S.~Ghemawat, ``{MapReduce}: Simplified data processing on large
  clusters,'' in \emph{USENIX OSDI}, Dec. 2004, pp. 137--150.

\bibitem{Schrage68}
L.~Schrage, ``A proof of the optimality of the shortest remaining processing
  time discipline,'' \emph{Operations Research}, vol.~16, pp. 687--690, 1968.

\bibitem{Smith78}
D.~R. Smith, ``A new proof of the optimality of the shortest remaining
  processing time discipline,'' \emph{Operations Research}, vol.~16, pp.
  197--199, 1978.

\bibitem{Jackson55}
J.~R. Jackson, ``Scheduling a production line to minimize maximum tardiness,''
  management Science Research Report, University of California, Los Angeles,
  CA, 1955.

\bibitem{Baccelli:1993}
F.~Baccelli, Z.~Liu, and D.~Towsley, ``Extremal scheduling of parallel
  processing with and without real-time constraints,'' \emph{J. ACM}, vol.~40,
  no.~5, pp. 1209--1237, Nov. 1993.

\bibitem{Smith56}
W.~E. Smith, ``Various optimizers for single-stage production,'' \emph{Naval
  Research Logistics Quarterly}, vol.~3, no. 1-2, pp. 59--66, 1956.

\bibitem{Rothkopf66}
M.~H. Rothkopf, ``Scheduling with random service times,'' \emph{Management
  Science}, vol.~12, no.~9, pp. 707--713, 1966.

\bibitem{Gittins79banditprocesses}
J.~C. Gittins, ``Bandit processes and dynamic allocation indices,''
  \emph{Journal of the Royal Statistical Society, Series B}, pp. 148--177,
  1979.

\bibitem{Klimov74}
G.~P. Klimov, ``Time-sharing service systems. {I},'' \emph{Theory of
  Probability \& Its Applications}, vol.~19, no.~3, pp. 532--551, 1975.

\bibitem{Gittins2011}
J.~Gittins, K.~Glazebrook, and R.~Weber, \emph{Multi-Armed Bandit Allocation
  Indices}, 2nd~ed.\hskip 1em plus 0.5em minus 0.4em\relax Wiley, 2011.

\bibitem{Leonardi:1997}
S.~Leonardi and D.~Raz, ``Approximating total flow time on parallel machines,''
  in \emph{ACM STOC}, 1997.

\bibitem{Weiss:1992}
G.~Weiss, ``Turnpike optimality of smith's rule in parallel machines stochastic
  scheduling,'' \emph{Math. Oper. Res.}, vol.~17, no.~2, pp. 255--270, May
  1992.

\bibitem{Weiss:1995}
------, ``On almost optimal priority rules for preemptive scheduling of
  stochastic jobs on parallel machines,'' \emph{Advances in Applied
  Probability}, vol.~27, no.~3, pp. 821--839, 1995.

\bibitem{Ying2015}
L.~Ying, R.~Srikant, and X.~Kang, ``The power of slightly more than one sample
  in randomized load balancing,'' in \emph{2015 IEEE INFOCOM}, April 2015, pp.
  1131--1139.

\bibitem{Dacre1999}
M.~Dacre, K.~Glazebrook, and J.~Niño-Mora, ``The achievable region approach to
  the optimal control of stochastic systems,'' \emph{Journal of the Royal
  Statistical Society: Series B (Statistical Methodology)}, vol.~61, no.~4, pp.
  747--791, 1999.

\bibitem{Stolyar_heavy2004}
A.~L. Stolyar, ``Maxweight scheduling in a generalized switch: State space
  collapse and workload minimization in heavy traffic,'' \emph{The Annals of
  Applied Probability}, vol.~14, no.~1, pp. 1--53, 2004.

\bibitem{Marshall2011}
A.~W. Marshall, I.~Olkin, and B.~C. Arnold, \emph{Inequalities: Theory of
  Majorization and Its Applications}, 2nd~ed.\hskip 1em plus 0.5em minus
  0.4em\relax Springer, 2011.

\bibitem{Sparrow:2013}
K.~Ousterhout, P.~Wendell, M.~Zaharia, and I.~Stoica, ``Sparrow: Distributed,
  low latency scheduling,'' in \emph{ACM SOSP}, 2013, pp. 69--84.

\bibitem{Borg2015}
A.~Verma, L.~Pedrosa, M.~R. Korupolu, D.~Oppenheimer, E.~Tune, and J.~Wilkes,
  ``Large-scale cluster management at {Google} with {Borg},'' in
  \emph{EuroSys}, Bordeaux, France, 2015.

\bibitem{michael2012book}
M.~L. Pinedo, \emph{Scheduling: Theory, Algorithms, and Systems}, 4th~ed.\hskip
  1em plus 0.5em minus 0.4em\relax Springer, 2012.

\bibitem{Yin_report2016}
\BIBentryALTinterwordspacing
Y.~Sun, C.~E. Koksal, and N.~B. Shroff, ``On delay-optimal scheduling in
  queueing systems with replications.'' [Online]. Available:
  \url{http://arxiv.org/abs/1603.07322}
\BIBentrySTDinterwordspacing

\bibitem{Liu1995}
Z.~Liu, P.~Nain, and D.~Towsley, ``Sample path methods in the control of
  queues,'' \emph{Queueing Systems}, vol.~21, no.~3, pp. 293--335.

\bibitem{Leonard_Kleinrock_book}
L.~Kleinrock, \emph{Queueing Systems}.\hskip 1em plus 0.5em minus 0.4em\relax
  John Wiley and Sons, 1975, vol. 1\& 2.

\bibitem{Jose2010}
J.~Nino-Mora, ``Conservation laws and related applications,'' in \emph{Wiley
  Encyclopedia of Operations Research and Management Science}.\hskip 1em plus
  0.5em minus 0.4em\relax John Wiley \& Sons, Inc., 2010.

\bibitem{Gittins:11}
J.~C. Gittins, K.~Glazebrook, and R.~Weber, \emph{Multi-armed Bandit Allocation
  Indices}, 2nd~ed.\hskip 1em plus 0.5em minus 0.4em\relax Wiley, Chichester,
  NY, 2011.

\bibitem{Chang93rearrangement_majorization}
C.-S. Chang and D.~D. Yao, ``Rearrangement, majorization and stochastic
  scheduling,'' \emph{Math. of Oper. Res}, 1993.

\bibitem{StochasticOrderBook}
M.~Shaked and J.~G. Shanthikumar, \emph{Stochastic Orders}.\hskip 1em plus
  0.5em minus 0.4em\relax Springer, 2007.

\end{thebibliography}

\end{document}